\documentclass[10pt,web]{ieeecolor}
\usepackage[outdir=./]{epstopdf}
\usepackage{color}
\usepackage[utf8]{inputenc}
\usepackage{amssymb,amsmath,mathtools}
\usepackage{breqn}
\usepackage{generic}
\usepackage{tensor}
\usepackage{braket}
\usepackage{graphicx}
\usepackage{soul}
\usepackage{comment}
\usepackage{float}
\usepackage{cancel}
\usepackage{siunitx}
\usepackage{flushend}
\allowdisplaybreaks
\graphicspath{{images/},{photos/}}
\usepackage[switch, mathlines]{lineno}
\newtheorem{theorem}{Theorem}
\newtheorem{corollary}{Corollary}

\newtheorem{remark}{Remark}
\def\Re{\mathop{\rm Re}}
\def\Im{\mathop{\rm Im}}
\def\IN{\text{IN}}
\def\OUT{\text{OUT}}
\def\ss{\text{ss}}
\DeclareMathOperator{\Tr}{Tr}

\title{Time-Domain Sensitivity of the Tracking Error}
\author{S.~O'Neil, \IEEEmembership{Member, IEEE},
  S.~G.~Schirmer, \IEEEmembership{Member, IEEE},
  F.~C.~Langbein, \IEEEmembership{Member, IEEE},
  C.~A.~Weidner, \IEEEmembership{Member, IEEE}, and
  E.~Jonckheere, \IEEEmembership{Life Fellow, IEEE}
  \thanks{This research was supported in part by NSF Grant IRES-1829078. SON is supported by the US Army Advanced Civil Schooling program.}
  \thanks{SON and EAJ are with the Department of Electrical and Computer Engineering, University of Southern California, Los Angeles, CA 90089 USA (email: \texttt{seanonei@usc.edu}, \texttt{jonckhee@usc.edu}).}
  \thanks{SGS is with the Faculty of Science \& Engineering, Swansea University, Swansea SA2 8PP, UK (e-mail: \texttt{s.m.shermer@gmail.com}).}
  \thanks{FCL is with the School of Computer Science and Informatics, Cardiff University, Cardiff CF24 4AG, UK (e-mail: \texttt{frank@langbein.org}).}
  \thanks{CAW is with the Quantum Engineering Technology Laboratories, University of Bristol, Bristol BS8 1FD, United Kingdom (e-mail: \texttt{c.weidner@bristol.ac.uk}).}
}

\begin{document}
\maketitle

\begin{abstract}
A strictly time-domain formulation of the log-sensitivity of the error signal to structured plant uncertainty is presented and analyzed through simple but representative classical and quantum systems. Results demonstrate that across a wide range of physical systems, maximization of performance (minimization of the error signal) asymptotically or at a specific time comes at the cost of increased log-sensitivity, implying a time-domain constraint analogous to the frequency-domain identity $\mathbf{S(s) + T(s) = I}$. While of limited value in classical problems based on asymptotic stabilization or tracking, such a time-domain formulation is valuable in assessing the reduced robustness cost concomitant with high-fidelity quantum control schemes predicated on time-based performance measures. 
\end{abstract}

\section{Introduction}

In the realm of feedback control, traditional sensitivity analysis of a closed-loop system to uncertain parameters is accomplished in the frequency-domain. Standard definitions for the sensitivity examine the derivative of the closed-loop plant $T(s)$ to differential perturbations in a given element $K(s)$ given by $\partial{T(s)}/\partial{K(s)}$. As this measurement scales with the units used to describe the plant and parameter, a more useful formulation is the \emph{log-sensitivity} of the closed-loop plant to variations in a given element through~\cite{Dorf2011}
\begin{equation}
\frac{\partial{T(s)}/T(s)}{\partial{K(s)}/K(s)} = \frac{\partial{T(s)}}{\partial{K(s)}}\frac{K(s)}{T(s)}.
\end{equation}
While valuable from a frequency-domain perspective, this method does not yield information about how the log-sensitivity evolves with time, with time-domain considerations often being grouped into performance measures such as rise and settling times. 

Some researchers have proposed methods for analyzing the sensitivity of system performance in the time-domain, though the methods tend to be system-specific. In~\cite{Chen2016} and~\cite{Dobes2005}, methods for analyzing the sensitivity of the output transient response of distributed transmission lines and microwave circuits are proposed. Additionally,~\cite{Nouri2018} and~\cite{Lombardi2019} provide methods for computing time-domain sensitivity measures for active and passive circuits. In particular,~\cite{Lombardi2019} convincingly demonstrates the computational efficiency of analytic methods over brute-force comprehensive perturbation analysis. While providing valuable methods for computing sensitivity in the time-domain, the current research in this area does not provide a predictive model relating sensitivity to performance metrics. This requirement for a predictive, time-domain model to gauge trade-offs in robustness and performance is becoming increasingly important in the field of quantum technology. Control problems in this field ranging from fast state transfer to the implementation of quantum gates are fundamentally time-based and not well-described by existing frequency-domain methods~\cite{Schirmer2015, Schirmer2016, Glaser2015}. Furthermore, the eigenstructure of closed quantum systems is characterized by poles on the imaginary axis that preclude application of common small-gain theorem-based robustness analysis methods such as structured singular value analysis~\cite{zhou:robust,10.5555/225507}.

In this paper, we extend the concept of the log-sensitivity from the frequency-domain analysis of transfer functions to the time-domain analysis of a signal. In particular, we examine the error signal $e(t) = y(t) - r(t)$ of a Single-Input, Single-Output (SISO) system to structured uncertainty in the system parameters. We first demonstrate the methodology with two classical systems and then extend the concept to quantum systems where time-domain specifications, particularly read-out time (i.e., the time at which the state of the system is measured), are crucial to system performance~\cite{walraff, simmons, putterman}. The main contribution of this paper is to provide a characterization of how the log-sensitivity of the error behaves as the output approaches the desired reference input. We show that the log-sensitivity of the error diverges to infinity as $y(t) \rightarrow r(t)$. Furthermore, the manner in which the log-sensitivity diverges is characterized by the multiplicity and character (real versus complex) of the dominant eigenvalue(s) of the closed-loop system and structure of the uncertain parameters.

In Section~\ref{prelim}, we establish the paradigm for calculation of the log-sensitivity of the error in terms of a classical SISO system with full-state feedback. Here, the pole-placement simultaneously meets design specifications and provides zero-steady state error as in~\cite{Chen2013}. In Section~\ref{math}, we derive the time-domain log-sensitivity of the error, prove that the limit of the log-sensitivity diverges as the output approaches the desired steady-state value, and characterize this divergence in terms of the dominant eigenvalue(s) of the closed-loop system. In Sections~\ref{spring-mass},~\ref{circuit},~\ref{qubit}, and~\ref{state transfer} we apply our analysis to both classical feedback systems and quantum systems, one subject to dissipation and one that evolves unitarily. This latter case is particularly interesting due to the difficulty of applying classical robust control methods to closed quantum systems, save for some special cases~\cite{IRP_quantum_survey_as_published, Petersen2013,IRP_old}.

\section{Preliminaries}\label{prelim}

We consider the general case of a SISO system with multiple state variables and the control objective of tracking a step input with zero error. The system is represented by
\begin{equation}\begin{aligned}
  \dot{x} & = \tilde{A}x + bu,\\
  y & = cx.
\end{aligned}\end{equation}
Here, $c \in \mathbb{R}^{1 \times N}$ and $b \in \mathbb{R}^{N \times 1}$ since we consider a SISO system. The matrix $\tilde{A} \in \mathbb{R}^{N \times N}$ is given by $\tilde{A} = A_1 + S\xi_0 + S(\xi-\xi_0)$. The nominal dynamics matrix is $A_1+S\xi_0$ and $\xi \in [\xi_1 ,\xi_2]$ is an uncertain parameter with nominal value $\xi_0 \in [\xi_1,\xi_2]$. This uncertain parameter enters the dynamics additively through the matrix $S$. 

Assuming the system is controllable, we use state feedback to place the poles of the system in accordance with our design specifications. Introducing the unit step reference signal $r(t)$, the system input is $u(t) = -kx(t) + k_0 r(t)$ where $k \in \mathbb{R}^{1 \times N}$ is the vector of static feedback gains and $k_0$ is the scalar gain used to scale the reference. Including the state feedback, we have the closed-loop state matrix $A = \tilde{A} - bk = A_1 - bk + S\xi$ and the state equation becomes $\dot{x} = Ax + bk_0r(t)$. The nominal state matrix with feedback is now $ A_0 = \left( A_1 - bk \right) + S\xi_0$.

We determine the time-domain state evolution as 
\begin{equation}\begin{aligned}
  x(t) & = e^{At}x(0) + \int_{0}^{t}e^{A(t-\tau)} bk_0r(\tau) \;d\tau\\
       & = e^{At}x(0) + k_0 e^{At} \int_{0}^{t} e^{-A \tau} b \;d\tau,
\end{aligned}\end{equation}
since $r(t)$ has unit magnitude. Without loss of generality and to simplify the exposition, we set $x(0) = 0$. Constraining our analysis to the zero-state response gives
\begin{equation}
  x(t) = k_0 (e^{At} - I)A^{-1}b.
\end{equation}
The term $-A^{-1}b$ enters as the vector of steady-state values of the step input-to-state response of the transfer function $(sI-A)^{-1}b = G(s)$. This follows immediately from evaluation of $G(s)|_{s=0} = -A^{-1}b$.

Since our goal is to track a unity step input so that $y(t) = cx(t) = r(t) = 1$, 
we ultimately want $-k_0 c A^{-1}b = 1$ or $k_0 = -(c A^{-1} b)^{-1}$. For simplicity, we write $A^{-1}b$ as the vector $\beta$. The output becomes
\begin{equation}
  y(t) = cx(t) = k_0 ce^{At}\beta - k_0 c\beta = k_0 c e^{At}\beta + 1,
\end{equation}
and we define the error signal as 
\begin{equation} \label{eq:time_domain_err}
  e(t) = r(t) - y(t) = -k_0 c e^{At} \beta.
\end{equation}
 
\section{Log Sensitivity of the Error}\label{math} 

With the time-domain error signal in~\eqref{eq:time_domain_err}, we define the log-sensitivity of the error to differential perturbations in the parameter $\xi$ as
\begin{equation}
  s(\xi_0,t) = \left.\frac{\partial e(t)}{\partial \xi} \frac{\xi}{e(t)} \right|_{\xi = \xi_0}.
\end{equation}
 
In general, the matrices $A = A_1 - bk +S\xi$ and $S$ do not commute. As such, calculation of the derivative of $e(t)$ with respect to the uncertain parameter $\xi$ follows from~\cite{Havel1995} where
\begin{equation}\label{eq: matrix_deriv}
\begin{aligned}
  \frac{\partial{e(t)}}{\partial{\xi}} 
  &= -k_0 c \frac{\partial}{\partial \xi} e^{(
  A_1 - bk +S\xi)t} \beta\\
  &= -k_0 c \left(\int_{0}^{t} e^{(t- \tau)A_0}S e^{\tau A_0} d \tau \right) \beta.
\end{aligned}
\end{equation}
 
To be precise, note that in the limit as $\Delta \xi \rightarrow 0$, we define the directional derivative of $e^{A_0t}$ in the direction of $S$ as in~\cite{Havel1995},
\begin{align}\label{eq:limit}
  D_S(t,A) 
   &= \lim_{\Delta \xi \rightarrow 0} \frac{1}{\Delta \xi} 
      (e^{t(A_1 - bk +S(\xi_0 + \Delta \xi))} - e^{t(A_1 - bk +S \xi_0)})\nonumber\\
   &= \lim_{\Delta \xi \rightarrow 0} \frac{1}{\Delta \xi} 
      {\left(e^{t\left(A_{0} + S\Delta \xi\right)} - e^{tA_{0}} \right),}
\end{align}
where $\Delta \xi=\xi-\xi_0$. 
As shown in the Appendix, we can express $\frac{\partial}{\partial \xi} e^{At}= MX(t)M^{-1}$ with $X(t)$ defined in Eq.~\eqref{eq:diag} or~\eqref{eq:non-diag}. Here, $M$ is the matrix of (generalized) eigenvectors that induce the similarity transformation $A_0 = MJM^{-1}$ with $J$ the Jordan normal form of $A_0$. 
Thus, 
\begin{equation}\label{eq: derivative_e(t)}
  \frac{\partial e(t)}{\partial{\xi}} = -k_0 c M X(t) M^{-1} \beta.
\end{equation}
Dividing by $e(t)$ while multiplying by $\xi_0$, we produce the log-sensitivity of the error
\begin{equation}\label{eq:log-sensitivity}
  \begin{aligned}
    s(\xi_0,t) &= \left.\frac{\partial e(t)}{\partial \xi} \frac{\xi}{e(t)}\right|_{\xi = \xi_0}
               = \frac{\xi_0 c M X(t) M^{-1} \beta}{c M e^{J t} M^{-1} \beta}.
  \end{aligned}
\end{equation}

Now consider the log-sensitivity in the case of perfect tracking when $y(t\to \infty) \to 1$ (equivalently $e(t \to \infty) \to 0$). Given a controllable linear system, by state feedback we guarantee convergence of $e(t)$ to zero, but only asymptotically as $t \rightarrow \infty$.

To determine whether the limit of $s(\xi_0,t)$ as $t \rightarrow \infty$ exists or if it diverges, we examine the ratio of the numerator $\mathcal{N}(t)$ and denominator $\mathcal{D}(t)$ of the scalar $s(\xi_0,t)$. Our expression for the log-sensitivity is now
\begin{equation}\label{eq: ratio}
   s(\xi_0,t)  
    = \left. \frac{\partial e(t)}{\partial \xi}
             \frac{\xi}{e(t)} \right|_{\xi=\xi_0}
    =  \frac{\mathcal{N}(t)}{\mathcal{D}(t)}.
\end{equation}
For simplicity we introduce the following notation.  Let $cM = z = \begin{bmatrix} z_1 & z_2 &  \ldots & z_N \end{bmatrix} \in \mathbb{C}^{1 \times N}$, 
where $z_k = cv_k$, the inner product of the row vector $c$ with the $k$-th (generalized) eigenvector of $A_0$. Likewise, describing the rows of $M^{-1}$ by $\nu_k$, we write the product $M^{-1}\beta$ as the column vector $w \in \mathbb{C}^{N \times 1}$ with components $w_k = \nu_k\beta$ and let $\bar{s}_{mn}$ be the elements of the matrix $\bar{S}=M^{-1}SM$.  Let the eigenvalues be ordered in increasing order of the magnitude of their real parts.  An eigenvalue $\lambda_m$ is \emph{dominant} if $\Re(\lambda_m)=\max_n\Re(\lambda_n)$.  Eigenvalues $\lambda_m$ with $\bar{s}_{mn}=\bar{s}_{nm}=0$ for all $n$ can be ignored. 

We now state the following main results of the paper:
\begin{theorem}\label{thm: linear_div}
If $A_0$ is diagonalizable with dominant, real eigenvalue $\lambda_1 \le 0$ with algebraic multiplicity one, then the log-sensitivity $s(\xi_0,t)=\xi_0 \bar{s}_{11} t + R(t)$, where $\lim_{t\to\infty} R(t)$ is finite, i.e., $s(\xi_0,t)$ diverges linearly as $\xi_0 \bar{s}_{11} t$ as $t \rightarrow \infty$ if $\bar{s}_{11}\neq 0$.
\end{theorem}
\begin{proof}
$\mathcal{N}(t)$ and $\mathcal{D}(t)$ in~\eqref{eq: ratio} take the form 
\begin{subequations}\begin{align}
  \mathcal{N}(t) &=  \xi_0 \sum_{m,n=1}^N z_m w_n \bar{s}_{mn} \phi_{mn}(t), \label{eq: N(t)} \\
  \mathcal{D}(t) &=\sum_{m=1}^N z_{m} w_{m} e^{\lambda_m t} \label{eq: D(t)},
\end{align}\end{subequations}
where 
\begin{equation}\label{eq: phi}
  \phi_{mn}(t) = \phi_{nm}(t) = \left\{\begin{matrix} 
    \frac{e^{\lambda_m t} - e^{\lambda_{n} t}}{\lambda_m - \lambda_{n}} & \text{for }\lambda_m \neq \lambda_n,\\[.5ex]
    t e^{\lambda_m t} & \text{for }\lambda_m = \lambda_n,
  \end{matrix} \right.
\end{equation}
and $\lambda_n$ are the eigenvalues of $A_{0} = A_1-bk+S\xi_0$, including repeated eigenvalues. Recall that all eigenvalues have $\Re(\lambda_n) \leq 0$, ensuring marginal stability at a minimum, and are ordered in increasing magnitude of their real parts so $\lambda_1$ is the dominant pole. Factoring $e^{\lambda_1 t}$ from both $\mathcal{N}(t)$ and $\mathcal{D}(t)$, and defining
\begin{equation}\label{eq: phi_tilde}
  \tilde{\phi}_{mn}(t) = \left\{\begin{matrix} 
    \frac{e^{(\lambda_m-\lambda_1) t} - e^{(\lambda_n -\lambda_1)t}}{\lambda_m - \lambda_n} & \text{for } \lambda_m \neq \lambda_n,\\[.5ex]
    t e^{(\lambda_m-\lambda_1) t} & \text{for } \lambda_m=\lambda_n,
  \end{matrix} \right.
\end{equation}
we write
\begin{equation}
    \frac{\mathcal{N}(t)}{e^{\lambda_1 t}} = \xi_0 \sum_{m,n=1}^N
    z_m w_n \bar{s}_{mn} \tilde{\phi}_{mn}(t).
\end{equation}
Noting that $\tilde{\phi}_{11}(t)=t$, $\tilde{\phi}_{m1}(t)$, $\tilde{\phi}_{1n}(t)$ contribute constant terms and all other $\tilde{\phi}_{m n}(t)$ only contribute terms that exponentially decay to $0$, we have
\begin{equation}\label{eq: N_case_1}
  \frac{\mathcal{N}(t)}{\xi_0 e^{\lambda_1t}}  
  = z_1 w_1 \bar{s}_{11}t + g_0 + \mathcal{N}_r(t)
\end{equation}
where 
\begin{equation}
g_0 = \sum_{n=2}^N \frac{z_1 w_n \bar{s}_{1n}}{\lambda_1 - \lambda_n}  - \sum_{m = 2}^N \frac{z_m w_1 \bar{s}_{m1}}{\lambda_m - \lambda_1} 
\end{equation}
and the terms in $\mathcal{N}(t)$ that decay to zero as $t \rightarrow \infty$ are
\begin{multline}
  \mathcal{N}_r(t) = 
  \sum_{m,n=2}^N z_m w_n \bar{s}_{mn} \tilde{\phi}_{mn}(t)
  - \sum_{n = 2}^N \frac{z_1 w_n \bar{s}_{1n}}{\lambda_1 - \lambda_n}e^{(\lambda_n - \lambda_1)t} \\
   + \sum_{m = 2}^N \frac{z_m w_1 \bar{s}_{m1}}{\lambda_m - \lambda_1} e^{(\lambda_m - \lambda_1)t}.
\end{multline}
For the denominator we have
\begin{equation}\label{eq: D_case_1}
  \frac{\mathcal{D}(t)}{e^{\lambda_1 t}} = z_1 w_1 + \sum_{m=2}^N z_m w_m e^{(\lambda_m - \lambda_1)t} = z_1 w_1 + \mathcal{D}_r(t),
\end{equation}
where $\mathcal{D}_r(t)$ likewise decays to zero. Now, for some $T > 0$, we have $|\mathcal{N}_r(t)| < N_0$ and $|\mathcal{D}_r(t)| < D_0$ where $N_0$ and $D_0$ are bounds at which the ratio of $N_0$ and $D_0$ to $z_1 w_1$ is negligible.  Finally, we have
\begin{equation}\label{eq:end_theorem_1}
\begin{split}
   s(\xi_0,t)  
   &= \frac{\xi_0
    \left(w_1 z_1 \bar{s}_{11} t + g_0 + \mathcal{N}_r(t)\right)
   }{z_1 w_1+\mathcal{D}_r(t)}\\
  &= \xi_0 \bar{s}_{11} t +R(t),
\end{split}
\end{equation}
where
\begin{equation}
  R(t) = \frac{\xi_0 \left(g_0 + \mathcal{N}_r(t)-\mathcal{D}_r(t)\bar{s}_{11}t\right)}{z_1w_1+\mathcal{D}_r(t)}.
\end{equation}
Since $\lim_{t\to \infty}R(t)$ is finite, if $\bar{s}_{11}\neq 0$ then $\xi_0\bar{s}_{11}t$ is the dominant term  of $s(\xi_0,t)$ as $t \to \infty$.
\end{proof}

\begin{corollary} If $A_0$ is diagonalizable with dominant, real eigenvalue $\lambda_1 \leq 0$ with equal algebraic and geometric multiplicity $\ell>1$, $s(\xi_0,t) = \xi_0 (a_0/b_0) t + R(t)$, where $R(t)$ remains finite, i.e., for $a_0\neq 0$, $s(\xi_0,t)$ again diverges linearly as $t \rightarrow \infty$ with slope $\xi_0 a_0/b_0$, where $a_0$ and $b_0$ are given by a linear combination of the coefficients associated with the dominant, repeated eigenvalue $\lambda_1$.
\end{corollary}

\begin{proof}
$\mathcal{N}(t)$ and $\mathcal{D}(t)$ follow from~\eqref{eq: N(t)} and~\eqref{eq: D(t)}.  Setting 
$a_0=\sum\limits_{m,n=1}^{\ell} z_m w_n \bar{s}_{mn}$, and
$b_0=\sum\limits_{m=1}^{\ell} z_m w_m$, we have
\begin{multline}\label{eq: diagonal_repeats}
  \mathcal{N}(t) = \xi_0 \left[a_0 t e^{\lambda_1 t} +  
  \sum\limits_{\substack{m = 1,n=\ell+1 \\ m=\ell+1,n=1}}^{N} z_m w_n \bar{s}_{m n} \phi_{mn}(t) \right],
\end{multline}
where the sum does not include repeats of the ordered pair $(m,n)$ and
\begin{equation}
  \mathcal{D}(t) =  b_0 e^{\lambda_1 t} +  \sum\limits_{m = \ell+1}^{N} z_m w_m e^{\lambda_{m} t}.
\end{equation}
Then~\eqref{eq:end_theorem_1} is modified as
\begin{equation}\label{eq:repeat}
    s(\xi_0,t) = \xi_0 (a_0/b_0) t + R(t), 
\end{equation}
where $R(t)$ again remains finite.
\end{proof}
\begin{remark}
Note that if $A_{0}$ has a real dominant eigenvalue $\lambda_1$ of matching algebraic and geometric multiplicity $m$, the theorem extends to repeated eigenvalues $\lambda_n \ne \lambda_1$  of arbitrary multiplicity.  Any terms that enter~\eqref{eq: N_case_1} and~\eqref{eq: D_case_1} generated by some $\lambda_n$ for $n\neq1$ necessarily have $\lvert \Re(\lambda_n) \rvert > \lvert \lambda_1 \rvert$ under the assumption of the dominant eigenvalue $\lambda_1$. If $\lambda_n$ is a simple root, upon factoring of $\lambda_1$ from $\mathcal{N}(t)$ and $\mathcal{D}(t)$, such terms are either constant or exhibit an exponential time dependence $e^{(\Re(\lambda_n)-\lambda_1)t} = e^{-\sigma_n t}$. 
If $\lambda_n$ is a repeated root with non-trivial Jordan block of multiplicity $\ell$ these terms take the form $ t^{\ell-1}e^{-\sigma_n t}$. In either case, such terms $\rightarrow 0$ as $t \rightarrow \infty$, are subsumed in $\mathcal{N}_r(t)$ and $\mathcal{D}_r(t)$ and grouped into $R(t)$. Any constant terms contributing to $\mathcal{N}(t)$, remain finite as $t \rightarrow \infty$ and are also included in $R(t)$. The result of the theorem is thus unchanged. 
\end{remark}

We now consider damped, complex conjugate eigenvalues. In the remainder of the paper $j$ is the imaginary unit. 
\begin{theorem}\label{thm: periodic}
If the dominant eigenvalue of $A_0 = A_1 - bk + S \xi_0$ appears in a complex conjugate pair $- \sigma \pm j \omega$, $\sigma \geq 0$, then $s(\xi_0,t) = (t f(t) + g(t) +\mathcal{N}_r(t))/(h(t) + \mathcal{D}_r(t))$ where $f(t)$, $g(t)$, $h(t)$ are periodic functions with period $\pi/\omega$ and $\mathcal{N}_r(t)$, $\mathcal{D}_r(t) \rightarrow 0$ as $t \rightarrow \infty$ with rate given by $(\Re(\lambda_3) + \sigma)$. 
Thus, $s(\xi_0,t)$ has no limit as $t \rightarrow \infty$, periodically taking divergingly large local maxima and local minima.
\end{theorem}

\begin{proof} 
Following the same procedure as for Theorem~\ref{thm: linear_div}, denote the dominant complex eigenvalue pair as $\lambda_{1,2} = - \sigma \pm j\omega$.  Factoring the real part of the dominant pole-pair gives
\begin{multline}\label{eq:complex}
 \mathcal{N}(t)/ \xi_0 e^{- \sigma t} = t z_1 w_1 \bar{s}_{11} e^{-j\omega t} + t z_2 w_2 \bar{s}_{22}e^{j \omega t}
   \\ + \sum_{\substack{(m,n) \neq (1,1) \\ (m,n) \neq (2,2)}} z_m w_m \bar{s}_{mn} \hat{\phi}_{mn}(t),
\end{multline}
where 
\begin{equation}\label{eq: phi_hat}
  \hat{\phi}_{mn}(t) = \left\{\begin{matrix} 
    \frac{e^{(\lambda_m+\sigma) t} - e^{(\lambda_n + \sigma)t}}{\lambda_m - \lambda_n} & \text{for } \lambda_m \neq \lambda_n, \\[.5ex]
    t e^{(\lambda_m+\sigma) t} & \text{for } \lambda_m=\lambda_n.
  \end{matrix} \right.
\end{equation}
The last term in~\eqref{eq:complex} generates $2N - 4$ terms of the form 
\begin{equation}
\frac{z_{(1,2)} w_n \bar{s}_{(1,2),n}}{\lambda_{(1,2)} - \lambda_n} ( e^{\pm j\omega t} - e^{(\lambda_n + \sigma)t})
\end{equation}
and $2N-4$ terms of the form 
\begin{equation}
\frac{z_{m} w_{(1,2)} \bar{s}_{m,(1,2)}}{\lambda_{m} -  \lambda_{(1,2)} } (  e^{(\lambda_m + \sigma)t} - e^{\pm j\omega t})
\end{equation}
along with the pair 
\begin{equation}
\frac{z_1 w_2 \bar{s}_{12}}{\lambda_1 - \lambda_2}(e^{-j \omega t} -e^{j \omega t}), \quad \frac{z_2 w_1 \bar{s}_{21}}{\lambda_2 - \lambda_1} (e^{j \omega t} - e^{- j \omega t}).
\end{equation}
Recalling that 
$\Re(\lambda_\ell) <-\sigma$, 
$\forall \ell \neq 1,2$, we rewrite the last term of Eq.~\eqref{eq:complex} as
\begin{multline}\label{eq: complex_residue}
   \left( \frac{z_1 w_2 \bar{s}_{12}}{\lambda_1 - \lambda_2} 
  - \frac{z_2 w_1 \bar{s}_{21}}{\lambda_2 - \lambda_1} \right)  (e^{- j \omega t} - e^{j \omega t}) + \\ 
   \left( \sum_{n = 3}^N \frac{z_1 w_n \bar{s}_{1n}}{\lambda_1 - \lambda_n}
  - \sum_{m = 3}^N \frac{z_m w_1 \bar{s}_{m1}}{\lambda_m - \lambda_1}\right) e^{-j \omega t} + \\   \left( \sum_{n = 3}^N \frac{z_2 w_n \bar{s}_{2n}}{\lambda_2 - \lambda_n} - \sum_{m=3}^N \frac{z_m w_2 \bar{s}_{m2}}{\lambda_m - \lambda_2} \right)e^{j \omega t} \mathcal + {N}_r(t),
\end{multline}
where $\mathcal{N}_r(t)$ contains all terms that decay to zero as $t \rightarrow \infty$.   
Recalling that the eigenvalues are ordered in decreasing value of $\Re(\lambda_k)$, we note that the dominant term in $\mathcal{N}_r(t)$ goes to zero as $e^{(\lambda_3 + \sigma)t}$. Regrouping terms that do not decay to zero in~\eqref{eq:complex} and~\eqref{eq: complex_residue}, yields
\begin{multline}
  tz_1w_1\bar{s}_{11}e^{-j\omega t} +  tz_2w_2 \bar{s}_{22}e^{j \omega t} \\
   + \left( \frac{z_1 w_2 \bar{s}_{12}}{\lambda_1 - \lambda_2} - \frac{z_2 w_1 \bar{s}_{21}}{\lambda_2 - \lambda_1} + \sum_{n = 3}^N \frac{z_1 w_n \bar{s}_{1n}}{\lambda_1 - \lambda_n} \right. \\ 
  \left. -\sum_{m = 3}^N \frac{z_m w_1 \bar{s}_{m1}}{\lambda_m - \lambda_1} \right)  e^{-j \omega t} + \left( \frac{z_2 w_1 \bar{s}_{21}}{\lambda_2 - \lambda_1} -\frac{z_1 w_2 \bar{s}_{12}}{\lambda_1 - \lambda_2}   \right. \\ 
  +\left. \sum_{n = 3}^N \frac{z_2 w_n \bar{s}_{2n}}{\lambda_2 - \lambda_n} - \sum_{m=3}^N \frac{z_m w_2 \bar{s}_{m2}}{\lambda_m - \lambda_2} \right) e^{j \omega t} \\ = t z_1 w_1 \bar{s}_{11} e^{-j\omega t} + t z_2 w_2 \bar{s}_{22}e^{j \omega t} + Q e^{-j \omega t} + Re^{j \omega t} \\ 
  =: tf(t) + g(t).~~~~~~~~~~~~~~~~~~~~~~~~~~~~~~~~~~~~~~~~
\end{multline}

In the denominator, following the factoring of $e^{-\sigma t}$, we have
\begin{equation}\label{eq:complex2}\begin{aligned}
  \frac{\mathcal{D}(t)}{e^{-\sigma t}} 
  &= z_1 w_1 e^{-j \omega t} + z_2 w_2 e^{j \omega t} + \sum_{n = 3}^N z_n w_n e^{(\lambda_n + \sigma)t}\\
  &= z_1 w_1 e^{-j \omega t} + z_2 w_2 e^{j \omega t} + \mathcal{D}_r(t)\\
  &= h(t) + D_{r}(t).
\end{aligned}\end{equation}
Here, $\mathcal{D}_r(t)$ denotes the terms in denominator that go to zero exponentially with dominant term $e^{(\lambda_3 + \sigma)t}$.  For $N=2$, the denominator is given exactly by the expression in~\eqref{eq:complex2} with $\mathcal{D}_r(t) = 0$.

To bound $\left| \mathcal{D}(t) / e^{-\sigma t} \right|$, note that $|\mathcal{D}_r(t)|$ achieves its maximum for $t \in [0,\pi/\omega)$. Furthermore, the complex exponential terms in $h(t)$ will vary between $\pm 2|z_1 w_1|$. Thus, the maximum value the denominator can attain is $\left| (2|z_1w_1|) + \max_{t\in[0,\pi/\omega)}\mathcal{D}_r(t) \right|$.

To analyze the behavior of the ratio of $\mathcal{N}(t)$ to $\mathcal{D}(t)$, we must examine the periodic behavior of $\mathcal{D}(t)$. Since $z_m$ and $w_n$ are, in general, complex, we must find where $|\mathcal{D}(t)| = 0$ or is a minimum. This yields the following condition for $|\mathcal{D}(t)/e^{-\sigma t}|$ as an asymptotic minimum:
\begin{multline}
  |z_1 w_1|^2 + |z_2 w_2|^2 
  + 2 \Re(z_1 w_1 z_2^* w_2^*) \cos(2 \omega t) \\
  - 2 \Im(z_1 w_1 z_2^* w_2^*) \sin(2 \omega t) = 0.
\end{multline}
Recalling that $v_1$ and $v_2$ (the eigenvectors associated with the dominant complex pole pair) are complex conjugates, we have $z_1 = z_2^*$ and $w_1 = w_2^*$. Combining the trigonometric functions to a single cosine term yields the equivalent, simplified, condition for the minimum of $\mathcal{D}(t)$:
\begin{equation}
  \cos \left( 2 \omega t -  \phi_{01} \right) = -1,
\end{equation}
where 
\begin{equation}
  \phi_{01} := \begin{cases}
    \phi, & \Re(z_1 w_1 z_2^* w_2^*) > 0, \\
    \phi + \pi, & \Re(z_1 w_1 z_2^* w_2^*) < 0,
  \end{cases}
\end{equation}
and
\begin{equation}
\phi := \arctan\left(\frac{-\Im(z_1 w_1 z_2^* w_2^*)}{\Re(z_1 w_1 z_2^* w_2^*)}\right)
\end{equation}
with $\arctan$ \emph{defined} to be in the first or fourth quadrant. We thus, expect the denominator to approach zero cyclically with a period $T = \pi/\omega$. 

Thus, $|\mathcal{D}(t)|$ remains bounded from above, but approaches zero periodically, which produces ``spikes'' in the log-sensitivity characterized by 
\begin{equation}
  |s(\xi_0,t_n) | = \left| \frac{ \mathcal{N}(t_n)}{\mathcal{D}(t_n)} \right| = 
  \left| \frac{t_nf(t_n) + g(t_n) + \mathcal{N}_r(t_n)}{h(t_n) + \mathcal{D}_r(t_n)} \right|,
\end{equation}
where $t_n = t_0 + n \pi/\omega$ and $t_0$ is the first time at which $\mathcal{D}(t)$ achieves a local minimum. In the case of $N=2$, this is given exactly by $t_0 = (\pi + \phi_{01})/(2\omega)$.
\end{proof}
\begin{corollary}\label{cor: corollary_2} If $A_0$ contains $m$ eigenvalues of the form $\lambda_m = \sigma \pm j \omega_m$ with $\sigma = \min\limits_{\lambda_n}\lvert\Re{ \lambda_n \rvert}$ and $\omega_m = m\omega_0$ (i.e. the eigenfrequencies are commensurate) the $\omega$ of Theorem~\ref{thm: periodic} determining the periodic behavior of $s(\xi_0,t)$ is {$\omega_0$}.
$\blacksquare$
\end{corollary}
\begin{remark}
If the $\{ \omega_m \}$ of Corollary~\ref{cor: corollary_2} are rationally independent so that $\sum_{n}\beta_n\omega_n = 0 \Rightarrow \beta_n = 0, \ \forall n$ the quasi-periodic behavior of $s(\xi_0,t)$ of Theorem~\ref{thm: periodic} is non-trivial and determined by expansion of all purely oscillatory terms of~\eqref{eq:complex} in 
$\omega_m$.
\end{remark}
\begin{theorem}
If $A_0$ has algebraic multiplicity $\ell$ in the dominant eigenvalue $\lambda_1$ with geometric multiplicity $1$, $s(\xi_0,t)$ diverges as a polynomial $F(t) = \left( \sum\limits_{n = 0}^{\ell} f_{n}(t) t^n \right) / \left(\frac{z_1 w_\ell}{(\ell-1)!} \right)$ as $t \rightarrow \infty$. 
\end{theorem}
\begin{proof}
Calculating the components of $\mathcal{N}(t) = \xi_0 zX(t)w $ from the results of the Appendix yields the following:

Since $X_1(t)$ is identical to the $X(t)$ of a diagonalizible matrix with $\ell$ repeated eigenvalues, $z X_1(t) w = e^{\lambda_1 t} \left( ta_{-\ell+2}(t) + a_{-\ell+1}(t) \right)$ where $a_{-\ell+2}(t)$ and $a_{-\ell+1}(t)$ are given by the terms in parentheses of~\eqref{eq: diagonal_repeats} multiplying $t$ or not, respectively, after factoring of $e^{\lambda_1 t}$.

Taking the products $z X_2(t) w$ and $z X_3(t) z$ updates~\eqref{eq:B int} and~\eqref{eq:C int} to sums of scalar products with $\bar{s}_{mr}z\Pi_{ms}w = \bar{s}_{mr}z_{m}w_{s}$ and $\bar{s}_{q m}z\Pi_{p m} w = \bar{s}_{q m}z_{p}w_{m}$. So,
\begin{align}
    z X_2(t) w &= e^{\lambda_1 t} \sum_{m = 0}^\ell b_{m-\ell+1}(t)t^m,\\
    z X_3(t) w &= e^{\lambda_1 t} \sum_{m = 0}^\ell c_{m - \ell +1}(t)t^m,
\end{align}
where $b_m(t)$ and $c_m(t)$ are composed of the terms in~\eqref{eq:B int} and~\eqref{eq:C int} grouped by like powers in $t$ after factoring of $e^{\lambda_1 t}$. Note that the largest power of $t$ in the polynomials $zX_2(t)w$ and $zX_3(t)w$ is $\ell$. Moreover,
\begin{equation}
    z X_4(t) w = e^{\lambda_1 t} \sum_{m = 3}^{2\ell-1} d_{m - \ell +1}(t)t^m,
\end{equation}
where $d_m(t)$ consists of those terms in powers of $t^m$ following the factoring of $\lambda_1$. As such,
\begin{multline}
  \frac{\mathcal{N}(t)}{\xi_0 e^{\lambda_1 t}} =  z\left( X_4(t) + X_3(t) + X_2(t) + X_1(t) \right) w \\
  =\sum_{m = 0}^{2\ell -1} t^m ( d_{m - \ell +1}(t) + c_{m-\ell+1}(t) + b_{m-\ell+1}(t) +  \\
  a_{m - \ell+1}(t) )  = \sum_{m=0}^{2\ell-1} t^m f_{m-\ell+1}(t)
\end{multline}
Note that for $m > \ell$, $f_m(t) = d_m(t)$. Also, since $\Re(\lambda_n) <\Re(\lambda_1) \le 0$,
$\forall n \neq 1$, $f_m(t)$ consists of two types of terms: (1) those that contain a factor of $e^{(\lambda_n - \lambda_1)t}$ and decay to zero and (2) terms, which are constant or purely oscillatory and do not decay to zero as $t \rightarrow \infty$.
    
The denominator has the more tractable expression
\begin{equation}
  \mathcal{D}(t) = \sum_{m=1}^N e^{\lambda_m t} z_m w_m + \sum\limits_{p=1}^{\ell-1}\sum\limits_{q=p+1}^\ell \frac{e^{\lambda_1 t}z_p w_q t^{(q-p)}}{(q-p)! }.
\end{equation}
The largest power of $t$ in $\mathcal{D}(t)$ is $\ell-1$ with coefficient $e^{\lambda_1 t}z_1 w_\ell/(\ell-1)!$.
Taking the ratio of $\mathcal{N}(t)/ \xi_0 \mathcal{D}(t)$ and cancelling common factors of $e^{\lambda_1 t}$ and $t^{\ell-1}$ yields
\begin{multline} \label{eq: polynomial_ratio}
 \frac{\sum\limits_{m=\ell-1}^{2\ell-1} f_{m - \ell +1}(t) t^{(m-\ell+1)} + \sum\limits_{m=0}^{\ell-2}f_{m-\ell+1}(t)t^{-(\ell-1-m)}}{ \left(\sum\limits_{n = 1}^N t^{1-\ell} z_n w_n e^{(\lambda_n - \lambda_1)t} +\sum\limits_{p=1}^{\ell-1} \sum\limits_{q=p+1}^\ell  \frac{z_p w_q}{(q-p)!}t^{(q-p+1-\ell)} \right) }\\ 
 = \frac{\sum\limits_{m=\ell-1}^{2\ell-1} f_{m-\ell+1}(t)t^{(m - \ell +1)} + \mathcal{O}(t^{-1})}{\frac{z_1 w_\ell}{(\ell-1)!} + \mathcal{O}(t^{-1})}
\end{multline}
where $\mathcal{O}(t^{-1}) \rightarrow 0 $ as $t \rightarrow \infty$ as $1/t$ or faster (i.e. with rate $t^{-n}$ or $e^{(\lambda_n - \lambda_1)t}$). Reindexing $m$ for clarity we have 
\begin{multline}
     s(\xi_0,t) = \xi_0 \frac{\mathcal{N}(t)}{\mathcal{D}(t)} = \xi_0 \frac{\sum\limits_{m=0}^\ell f_m(t)t^{m} + \mathcal{O}(t^{-1})}{\frac{z_1 w_\ell}{(\ell-1)!} + \mathcal{O}(t^{-1})}\\ = \xi_0 \frac{\sum\limits_{m=0}^{\ell}f_m(t) t^m}{\frac{z_1w_\ell}{(\ell - 1)!}} + R(t)
\end{multline}
where $R(t) \rightarrow 0$ as $t \rightarrow \infty$. Then, as $t \rightarrow \infty$, we have
\begin{equation}
  \frac{\mathcal{N}(t)}{\mathcal{D}(t)} \rightarrow \frac{\sum\limits_{m=0}^\ell f_m(t)t^m}{\frac{z_1 w_\ell}{(\ell-1)!}} = F(t),
\end{equation}
so that $s(\xi_0,t) = \xi_0 \mathcal{N}(t)/\mathcal{D}(t) \rightarrow \infty$ as a polynomial in $t^m$. Before concluding, note that we can lift the assumption of distinct eigenvalues for $\lambda_n \neq \lambda_1$. By assumption, $\lvert \Re \lambda_n \rvert > \lvert \Re \lambda_1 \rvert$ for all $n > 1$ so that any terms in~\eqref{eq: polynomial_ratio} generated by a Jordan block not associated with $\lambda_1$ decay as $e^{(\Re(\lambda_n) - \Re(\lambda_1))t}$ and are subsumed in $\mathcal{O}(t^{-1})$ leaving the result of the theorem unchanged. 
\end{proof}
\begin{remark}
Note that if the dominant eigenvalues are characterized by $\Re{\lambda_1} = 0$, the results of this section still hold. This is easily verified by noting that factoring of $e^{\lambda_1 t}$ does not change the leading terms of $\mathcal{N}(t)$ or $\mathcal{D}(t)$ and any remaining terms of the form $e^{\lambda_{k}t}$ for $k > 1$ decay to zero as $t \rightarrow \infty$ under the assumption of feedback stabilization from Section~\ref{prelim}.
\end{remark}

\section{Classical Example -- Spring-Mass System}\label{spring-mass}

We first examine the case of an undamped spring-mass system we wish to position at $x_{final} = \SI{1}{m}$ with an actuating force on the mass that provides the step-input. Taking the spring as the uncertain variable with nominal value of $k = \xi_0 = \SI{4}{N/m^2}$, the state-equation for the nominal system is: 
\begin{equation}
  \begin{bmatrix} \dot{x_1} \\ \dot{x_2} \end{bmatrix} = \begin{bmatrix}0 & 1 \\ -\xi & 0 \end{bmatrix} \begin{bmatrix} x_1 \\ x_2 \end{bmatrix} + \begin{bmatrix}0 \\ 1 \end{bmatrix} u.
\end{equation}
Here, $x_1$ is the mass position, and $x_2$ is the velocity. We choose $x_1$ as the measured output, so $c = \begin{bmatrix}1 & 0\end{bmatrix}$.

Variations about the nominal value of the spring constant enter the dynamics additively through the structure matrix as $(\Delta \xi) S$ where $S = \begin{bmatrix} 0 & 0 \\ 1 & 0 \end{bmatrix}$.

\subsection{Real Dominant Eigenvalue}

We first choose real, distinct eigenvalues for rapid convergence with no oscillation. Choosing $\lambda_1 = -2$ and $\lambda_2 = -5$ produces a step response with zero overshoot, rise time of $\SI{1.23}{s}$, and settling time of $\SI{2.21}{s}$. The resulting limiting behavior of $\left| s(\xi_0,t) \right|$ is shown in Fig.~\ref{fig:spring-mass1}. In accordance with Theorem~\ref{thm: linear_div}, the log-sensitivity diverges linearly with a slope given by $|\xi_0\bar{s}_{11}| = 4/3$.

\begin{figure}
\centering
\includegraphics[width=\columnwidth]{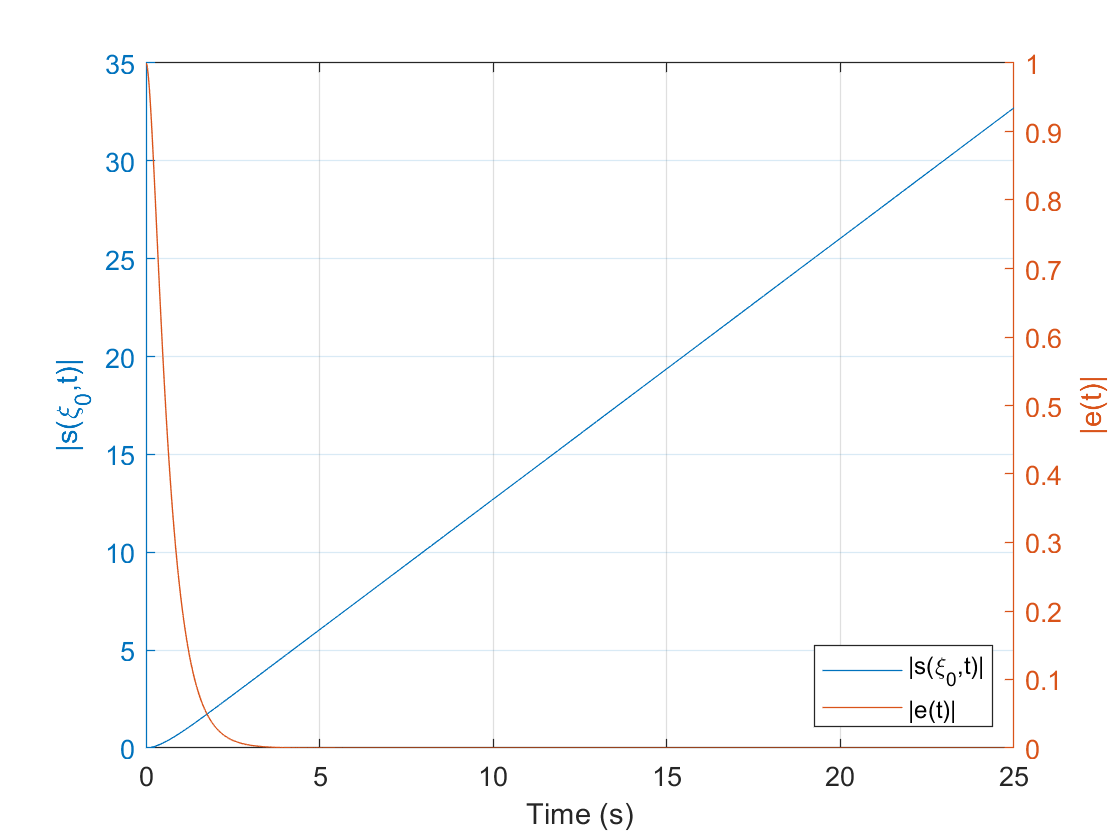}
\caption{Spring-mass system with $\lambda_1 = -2$, $\lambda_2=-5$, $\xi_0 = 4$, and $\bar{s}_{11} = -1/3$. Note the linear divergence of $\left| s(\xi_0,t) \right|$ with a slope of $4/3$.}
\label{fig:spring-mass1}
\end{figure}

\subsection{Complex Dominant Eigenvalue Pair}

We now choose eigenvalues of $-1 \pm j \pi/5$ to yield a system with lighter damping and oscillatory dynamics. This provides a more gentle response with an overshoot of $0.67$, rise time of $\SI{2.24}{s}$, and settling time of $\SI{3.52}{s}$. For the log-sensitivity of the error, we first note that $\Re({z_1 w_1 z_2^* w_2^*}) = -0.197$ and $\Im{(z_1 w_1 z_2^* w_2^*)} = -0.409$. Not considering the additional factor of $\pi$ in $\phi_{01}$ produces an erroneous first zero crossing time of $t_0 = \SI{1.61}{s}$. Taking into account the sign of $\Re({z_1 w_1 z_2^* w_2^*)}$ agrees with the expected periodic behavior. Specifically, $t_0 = (\pi + \phi_{01})/(2 \omega) = \SI{4.107}{s}$ with expected, asymptotic recurrence times of $t_n = t_0 + \left( \pi/\omega \right) n$ as stated in Theorem~\ref{thm: periodic}. The results are shown in Fig.~\ref{fig:spring-mass-complex-eig}. Note that the local maxima for $\left| s(\xi_0,t) \right|$ and local minima for $e(t)$ correspond to the values of $t_n$ predicted by Theorem~\ref{thm: periodic}. 

\begin{figure}
\centering
\includegraphics[width=\columnwidth]{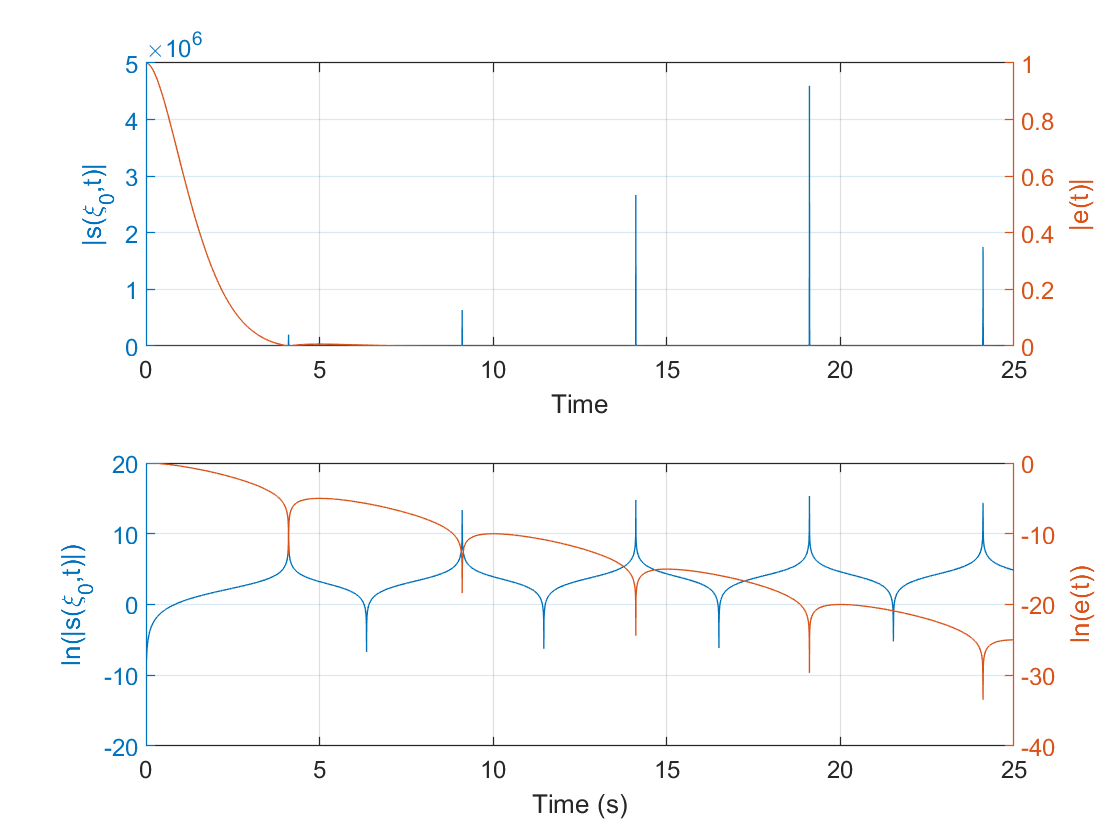}
\caption{Divergence of $\left|s(\xi_0,t)\right|$ for the spring-mass system with a complex eigenvalue pair at $s = -1 \pm j \pi/5$. The top plot shows both, $e(t)$ and $\left| s(\xi_0,t)\right|$, on a linear scale, and the bottom plot shows both on a log-scale. Note that $s(\xi_0,t)$ displays local maxima every $\pi/\omega = \SI{5}{s}$ as the error periodically goes to zero.}\label{fig:spring-mass-complex-eig}
\end{figure}

\section{Classical Example---RLC Circuit}\label{circuit}

Extending the procedure to a slightly more complex scenario, we consider an RLC circuit as depicted in Fig.~\ref{fig:circuit}.

\begin{figure}
\centering
\includegraphics[width=\columnwidth]{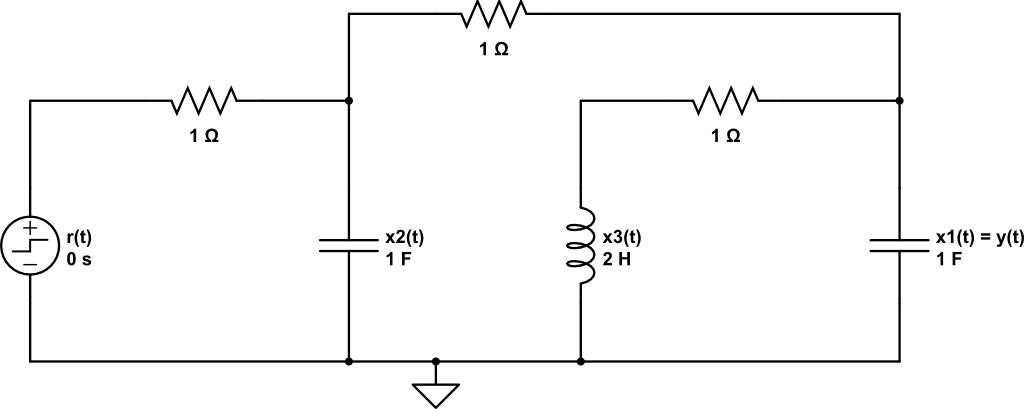}
\caption{RLC circuit with three states consisting of the two capacitor voltages and single inductor current. The input is a voltage step at $t=0$ and the output is the capacitor voltage $x_1(t)$ in the rightmost branch.}\label{fig:circuit}
\end{figure}

The voltage source provides a step input of $\SI{1}{V}$. The control objective is tracking this step input voltage at the capacitor voltage in the rightmost branch. The inductance in the system is the uncertain parameter with a nominal value of $\SI{2}{H}$. This provides the following state-space set-up
\begin{equation}
  \begin{bmatrix}\dot{x_1} \\ \dot{x_2} \\ \dot{x_3} \end{bmatrix} = \begin{bmatrix}-1 & 1 & -1 \\ 1 & -2 & 0 \\ \xi & 0 & -\xi \end{bmatrix} \begin{bmatrix} x_1 \\ x_2 \\ x_3 \end{bmatrix} + \begin{bmatrix} 0 \\ 1 \\ 0 \end{bmatrix} u.
\end{equation}
The nominal inverse inductance is $\xi_0 = 1/2$ and we have
\begin{equation}
  S = \begin{bmatrix} 0 & 0 & 0\\ 0 & 0 & 0 \\ 1 & 0 &-1\end{bmatrix}.
\end{equation}
The current through the inductor is taken as $x_3$ and the voltage across the output capacitor is $x_1$. Since we measure $x_1$ as the output, the output vector is $c = \begin{bmatrix} 1 & 0 & 0 \end{bmatrix}$.

\subsection{Real Dominant Eigenvalue}

We first consider real eigenvalues and a dominant eigenvalue of $\lambda_1 = -1$. The system response has a rapid rise time of $\SI{0.49}{s}$ and a large overshoot of $30 \%$. This overshoot is attributable to the negative residue of $\lambda_2 = -2$ which generates a non-monotonic convergence of $y(t)$ to the steady state $y_{\ss} = 1$.  The behavior of the log-sensitivity with time is shown in Fig~\ref{fig:circuit_real}. In accordance with Theorem~\ref{thm: linear_div}, the log-sensitivity diverges with slope $|\xi_0 \bar{s}_{11} | = \left| 0.5 (-3.167) \right| = 1.58$. Contrasted with this long-term behavior, we note a local maximum of $\left| s(t,\xi_o) \right|$ at $t = \SI{0.701}{ s}$ when the step response passes through $y(t) = 1$, attributable to the transient dynamics.

\begin{figure}
\centering
\includegraphics[width=\columnwidth]{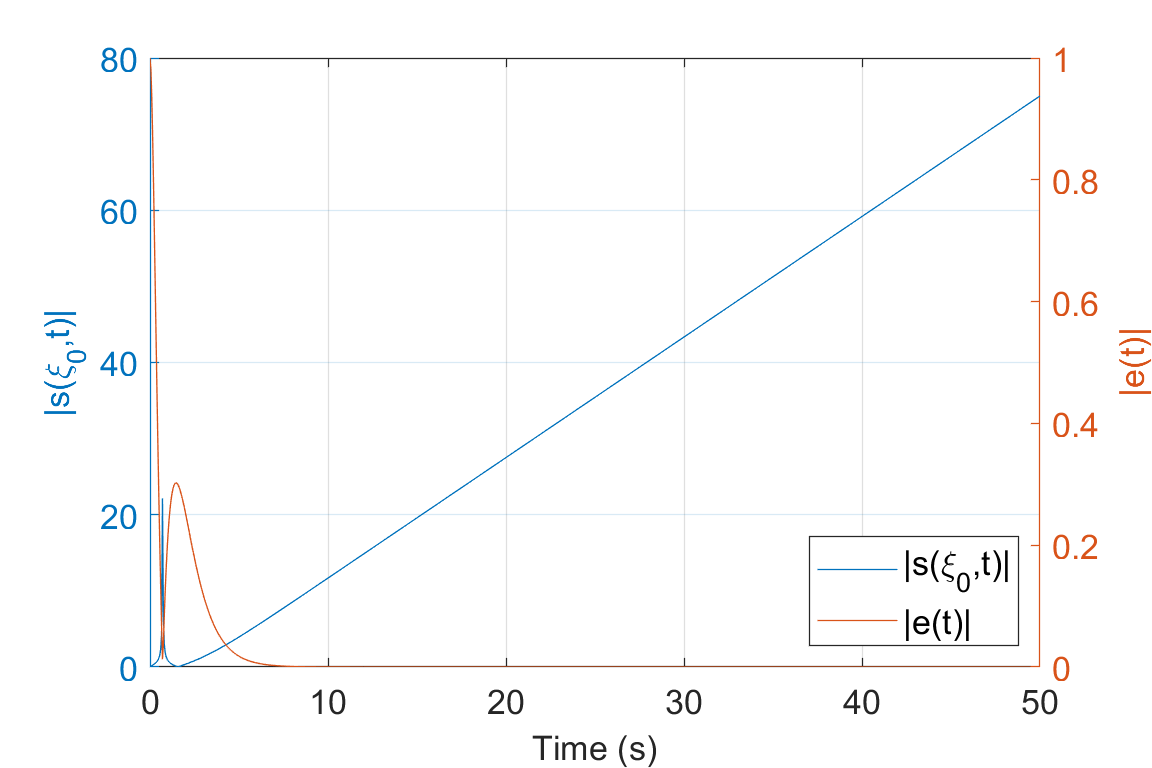}
\caption{Divergence of $\left|s(\xi_0,t)\right|$ for the third order circuit with $\lambda_1 = -1$, $\lambda_2 = -2$, and $\lambda_3 = -4$. As predicted, the log-sensitivity of the error diverges linearly with time.}\label{fig:circuit_real}
\end{figure}

\subsection{Complex Dominant Eigenvalue Pair}
\begin{figure}
\centering
\includegraphics[width=\columnwidth]{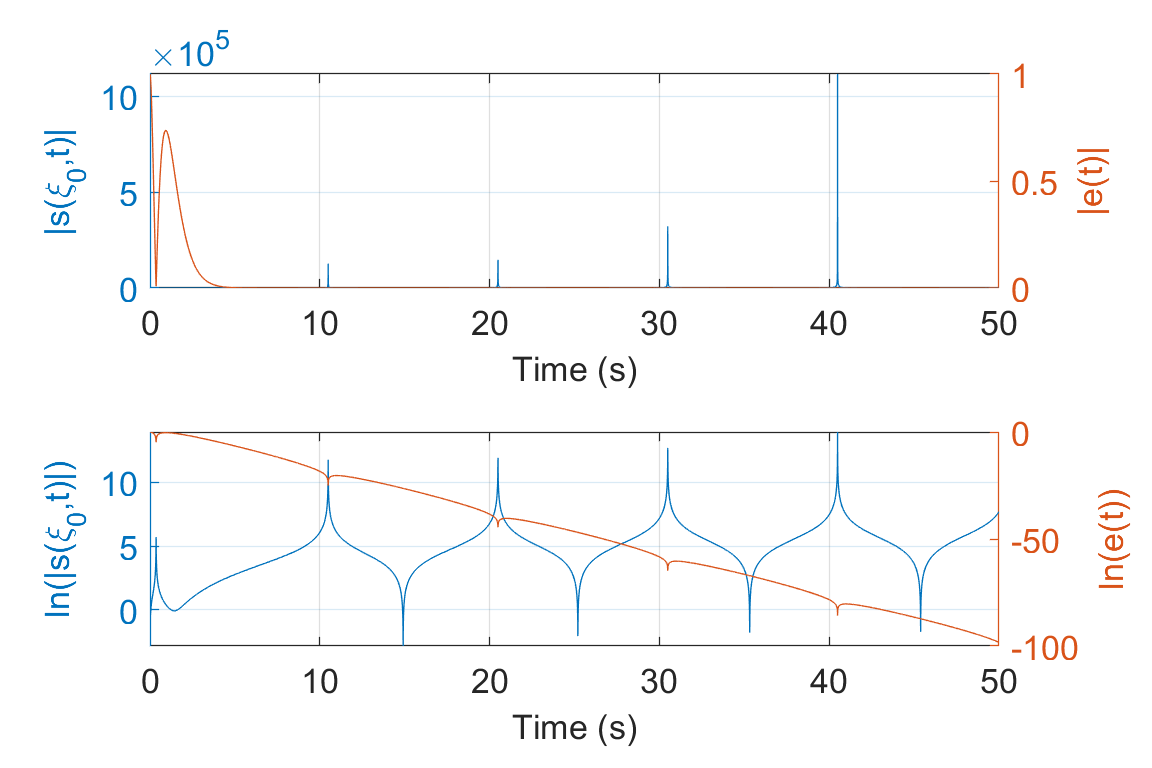}
\caption{$\left| s(\xi_0,t) \right|$ diverging over time for a third-order circuit with dominant complex eigenvalue pair $\lambda_{1,2} = -2 \pm j \pi/10$. The top panel displays $\left|s(\xi_0,t)\right|$ and $e(t)$ on a linear scale, and the lower panel displays the same on a log-scale. Note the periodic maxima of $\left|s(\xi_0,t)\right|$ and corresponding minima of $e(t)$ with period $\SI{10}{s}$.}\label{fig:circuit_complex}
\end{figure}

Next, we choose complex eigenvalues of $\lambda_{1,2} = -2 \pm j\pi/10$. As seen in Fig.~\ref{fig:circuit_complex}, $|s(\xi_0,t)|$ does not approach a limiting value as $t \rightarrow \infty$, but grows unbounded with periodic local maxima at a period of $\pi/\omega = \SI{10}{s}$, in accordance with Theorem~\ref{thm: periodic}. We also note a first spike in $|s(\xi_0,t)|$ at $t = \SI{0.350}{s}$ when $y(t)$ passes through $1$ during the transient response. The next local maximum occurs at the predicted time of $t_0 = (\pi + \phi_{01})/(2 \omega) = \SI{10.49}{s}$. The subsequent local maxima in $\left| s(\xi_0,t) \right|$ follow at the expected times of $t_n = \left( t_0 + (\pi/\omega) n \right) = (10.49 + 10n)\si{s}$.

\section{Open Quantum Systems Example -- Two Qubits in a Cavity}\label{qubit}

\subsection{System Description and Problem Formulation}

We now examine how the postulated long-term behavior of the log-sensitivity applies to non-classical systems. We consider a simple, open quantum system with a globally asymptotic steady state, as detailed in~\cite{Schirmer2022}. This asymptotic convergence facilitates similarity with the behavior of classical systems.

Consider two qubits (quantum mechanical two-level systems) collectively coupled to one another via a lossy cavity, as originally detailed in~\cite{motzoi}. As an open quantum system, the dynamics are governed by the time-dependent Liouville equation
\begin{equation}\label{eq:two-qubit}
    \frac{d}{dt} \rho(t) 
    = \left[ H,\rho(t) \right] + \mathcal{L}(V_\gamma) \rho(t),
\end{equation}
where the cavity has been adiabatically eliminated~\cite{motzoi}, $H$ is the Hamiltonian that determines the evolution of the system, $V_\gamma$ a constant dissipation operator, $\rho(t)$ the density operator that encodes the state information, and $[\cdot,\cdot]$ the matrix commutator. For this specific two-level system, we have~\cite{Schirmer2022}
\begin{equation}
  H = \begin{bmatrix} 0 & \alpha_2 & \alpha_1 & 0 \\ \alpha_2^* & \Delta_2 & 0 & \alpha_1 \\ \alpha_1^* & 0 & \Delta_1 & \alpha_2 \\ 0 & \alpha_1^* & \alpha_2^* & \Delta_1 + \Delta_2 \end{bmatrix}, \
  V_\gamma = \begin{bmatrix} 0 & \gamma_2 & \gamma_1 & 0 \\ 0 & 0 & 0 & \gamma_1 \\ 0 & 0 & 0 & \gamma_2 \\ 0 & 0 & 0 & 0 \end{bmatrix}.
\end{equation}
The terms $\alpha_1$ and $\alpha_2$ represent the driving fields of the qubits, and $\Delta_1$ and $\Delta_2$ represent the detuning parameters, i.e., the difference between the driving field frequency and the qubit resonance frequency for qubits $1$ and $2$, respectively. The terms $\gamma_1$ and $\gamma_2$ in the matrix $V_\gamma$ provide the strength of the decoherence acting on the first and second qubit, respectively. We take the nominal values of $\alpha_n$ and $\gamma_n$ as $1$, $\Delta_1$ as $-0.1$, and $\Delta_2$ as $0.1$.

We consider the following perturbations to these parameters in accordance with~\cite{Schirmer2022} and the associated structure matrices where $\delta_{mn}$ denotes a $4 \times 4$ matrix with a one in the $\left( m, n \right)$ location and zeros elsewhere:
\begin{itemize}
\item Perturbations to $\alpha_1$ with $S_1 =\delta_{13} + \delta_{31} + \delta_{24} + \delta_{42}$,
\item Perturbations to $\alpha_2$ with $S_2 = \delta_{12} + \delta_{21} + \delta_{34} + \delta_{43}$,
\item Perturbations to $\Delta_1$ with $S_3 = \delta_{33} + \delta_{44}$,
\item Perturbations to $\Delta_2$ with $S_4 = \delta_{22} + \delta_{44}$.
\end{itemize}
The equations of motion do not readily lend themselves to analysis in the common state-space formalism, but we can use the Bloch formulation to accomplish this. We choose an orthonormal basis of the $N \times N$ Hermitian matrices $\{\sigma_{n}\}$ where the first $N^2-1$ elements are traceless, and $\sigma_{N^2} = (1/\sqrt{N})I_{N^2}$ with $N$ the dimension of the system. We define
\begin{subequations}\begin{align}
  A_{mn} &= \Tr(jH[\sigma_m,\sigma_n]), \label{eq:bloch1}\\
  L_{mn} &= \Tr(V_\gamma^\dagger \sigma_m V_\gamma \sigma_n - \frac{1}{2} V_\gamma^\dagger V_\gamma \{\sigma_m,\sigma_n \} ), \label{eq:block2}\\
  r_m(t) &= \Tr(\sigma_m \rho(t)). \label{eq:bloch3}
\end{align}\end{subequations}
With $N=4$, this yields in $\dot{r}(t) = (A +L)r(t)$ with $A,L \in \mathbb{R}^{16 \times 16}$ and $r(t) \in \mathbb{R}^{16}$. Together with $r_0 = r(0)$, we have the standard equations that represent an autonomous state-space system with \emph{free response} $r(t) = e^{t(A+L)}r_0$.

The system has a single zero eigenvalue, and the nullspace of $A+L$ provides the steady-state associated with this zero eigenvalue, denoted as $r_{\ss}$. We define the output as the scalar $y(t) = r_{\ss}^{T}r(t)$ where $0 \leq y(t) \leq 1$, and $y(t)$ represents the overlap of the current state with the steady-state. We define the overlap error as $1 - y(t) = 1 - r_{\ss}^Tr(t)$. Noting that for any state $\rho(t)$, $r_{N^2} = \Tr\left( (1/\sqrt{N}) \rho(t)\right) = 1/\sqrt{N}$ as a consequence of the constancy of the trace for density matrices, we have $1 = r_1^Tr(t)$, where $r_1$ is a vector of all zeros save for the $N^2$-th entry, which is $\sqrt{N}$. We can then simplify the expression for the error as $(r_1 - r_{\ss})^Tr(t) = cr(t)$ where $c \in \mathbb{R}^{N^2 \times 1}$ consisting of $c_n = -r_{{\ss}_n}$ for $n=1, \dotsc, N^2-1$ and $c_{N^2} = \frac{N-1}{\sqrt{N}}$. Specifically for this case we have $c_{16} = \frac{3}{2}$.

Perturbations $S_1$ through $S_4$ map linearly to the Bloch formulation~\cite{PhysRevA.93.063424,PhysRevA.81.062306,JPhysA37,neat_formula} via~\eqref{eq:bloch1} to produce a structure matrix $\tilde{S} \in \mathbb{R}^{16 \times 16}$. Thus, a differential perturbation of the form $\Delta \xi S_n$ for  $n \in \{ 1, \ 2, \ 3, \ 4 \}$ in~\eqref{eq:two-qubit} maps to $\Delta\xi \tilde{S}$, and we have the following perturbed form of the time evolution of the overlap error:
\begin{equation}
  e(t) = ce^{t(A + L + \Delta \xi \tilde{S})}r_0.
\end{equation}
This allows us to compute the derivative of $e(t)$ with respect to perturbations in $\xi$ structured as $\tilde{S}$ in accordance with~\eqref{eq:limit} and~\eqref{eq:log-sensitivity}.

Before proceeding to the behavior of $s(\xi_0,t)$ we note:
\begin{enumerate}
\item In contrast to the two classical examples, there is no full-state feedback that modifies the dynamics of the system. The control is accomplished through the driving fields $\alpha_n$ and detuning $\Delta_n$ to modify the evolution of the state in an \emph{a priori} manner. 
\item As opposed to the classical case studies, we do not assume a zero initial state. The probability that the two-qubit ensemble is in \emph{some} state at $t=0$ requires a non-zero $\rho_0$ or equivalently $r_0 \neq 0$.
\end{enumerate}
Despite these differences, the mathematical form of $e(t)$ is identical to the classical formulation and amenable to the same results derived in Section~\ref{math}.

\subsection{Log-Sensitivity of the Error}

In accordance with~\eqref{eq:limit} and~\eqref{eq:log-sensitivity}, we calculate the derivative of the error to perturbations in $\alpha_n$ and $\Delta_n$ by 
\begin{equation}\begin{aligned}
  \frac{\partial e(t)}{\partial \xi} &=  \lim_{\Delta \xi \rightarrow 0} \frac{1}{\Delta\xi}c(e^{t(L+A+\Delta \xi \tilde{S})} - e^{t(A+L)})r_0\\
  &= D_{\tilde{S}}(t,A+L).
\end{aligned}\end{equation}
The two dominant eigenvalues of $A+L$ are $\lambda_1 = 0$, followed by a purely real eigenvalue of $\lambda_2 = -0.0035$. The $\bar{s}_{11}$ term is zero for all perturbations considered and does not contribute to the sum for $s(\xi_0,t)$. Thus, in accordance with Theorem~\ref{thm: linear_div}, we anticipate that the behavior of $s(\xi_0,t)$ is dominated by $\lambda_2$ and the associated structure term $\bar{s}_{22}$, and we expect the slope of the divergence to be equal to $\left|\xi_0 \bar{s}_{22} \right|$.

The result for a differential perturbation in the driving fields $\alpha_1$ or $\alpha_2$ is illustrated in Fig.~\ref{fig:qubit1}. With a nominal value of $\alpha_1 = \xi_0 = 1$ and $\bar{s}_{22} = 0.00344$, the predicted slope of $0.00344$ is borne out by Fig~\ref{fig:qubit1}. A perturbation to $\alpha_2$ with structure $S_2$ yields the same result.

\begin{figure}
\centering
\includegraphics[width=\columnwidth]{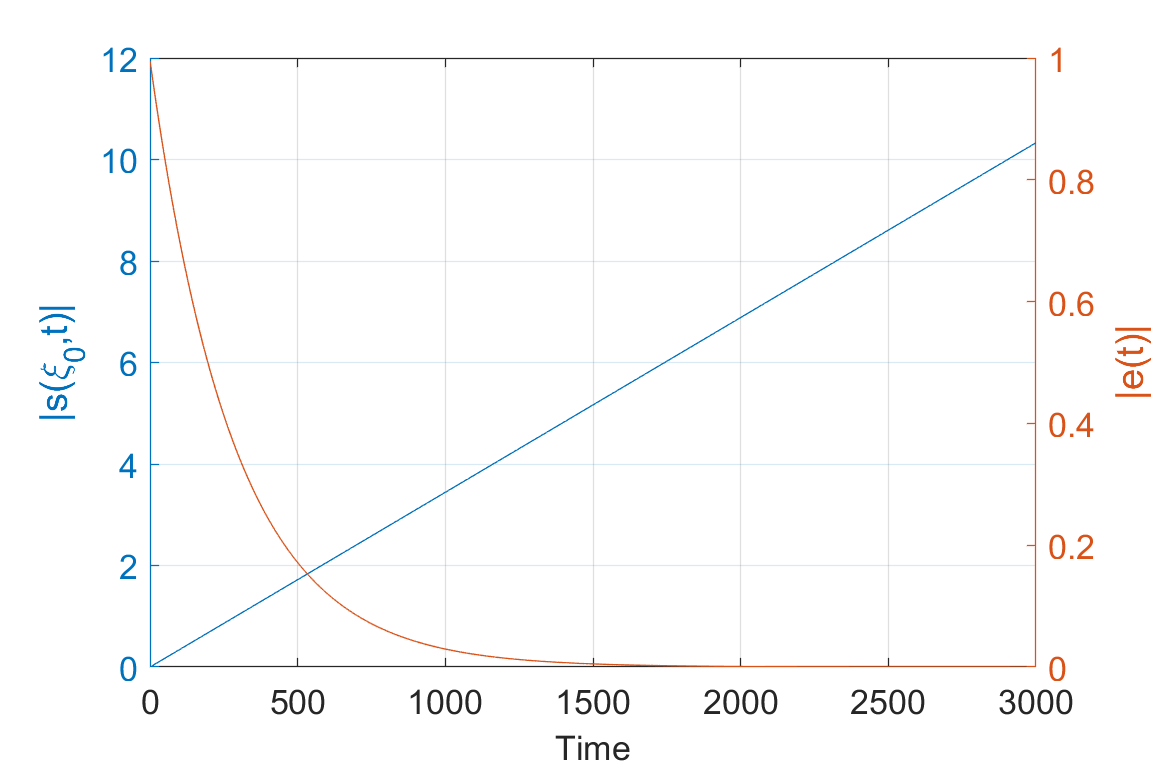}
\caption{Divergence of $\left| s(\xi_0,t) \right|$ for perturbations to the driving fields $\alpha_{1,2}$, with slope given by $\left| \xi_0 \bar{s}_{22} \right| = \left| (1)(0.00344) \right| = 0.00344$.}\label{fig:qubit1}
\end{figure}

The case for a perturbation to the detuning parameters is illustrated in Fig.~\ref{fig:qubit2}. As predicted by Theorem~\ref{thm: linear_div}, we observe a slope of $\left| \xi_0 \bar{s}_{22} \right| = \left| (\pm 0.1) (0.0351) \right| = 0.00351$.

\begin{figure}
\centering
\includegraphics[width=\columnwidth]{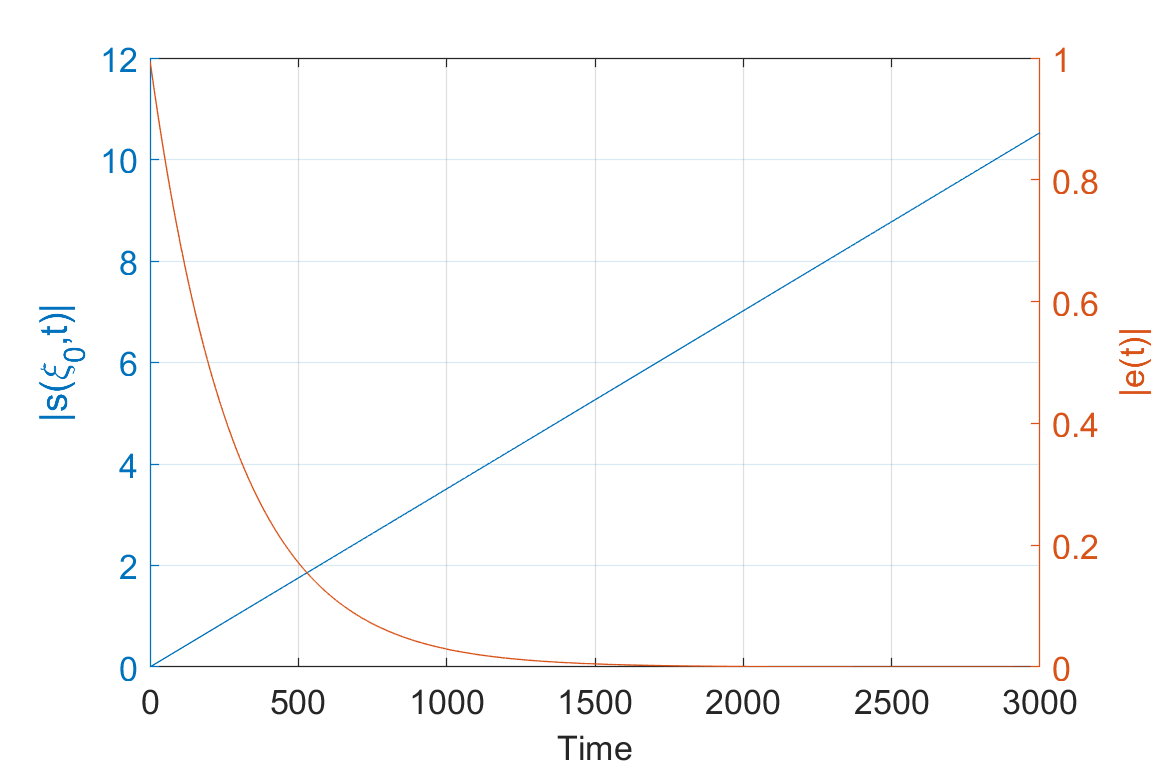}
\caption{Divergence of $ \left| s(\xi_0,t) \right|$ for perturbations to the detuning parameters $\Delta_{1,2}$ with slope given by $\left|\xi_0  \bar{s}_{22} \right| = \left| (-0.1)(-0.0351) \right| = 0.00351$.}\label{fig:qubit2}
\end{figure}

\section{Closed Quantum System Example -- Perfect State Transfer}\label{state transfer}
 
\subsection{System Description and Problem Formulation}

Consider a system designed to facilitate perfect state transfer in a chain of $N$ particles characterized by spins~\cite{Christandl2004,NEPOMECHIE1999}. Though the procedure is generally applicable to multiple excitations, for simplicity we restrict our attention to the case of transfer of a single excitation without dissipation. This is the so-called \emph{single excitation subspace}.
 
As in~\cite{Christandl2004}, we represent the state of the system as a column vector $\ket{\psi} \in \mathbb{C}^{N}$ with a one in the $n$-th entry to denote a single excitation that is associated with the $n$-th spin. The design goal is to transfer the single excitation from spin $1$ ($\ket{\psi_{\text{IN}}} = \begin{bmatrix} 1 & 0 & 0 & \cdots & 0 \end{bmatrix}^T$) to $N$ ($\ket{\psi_{\text{OUT}}} = \begin{bmatrix} 0 & 0 & \cdots & 1 \end{bmatrix}^T$) at a finite time $T = \pi/\lambda$. Here, $\lambda$ is a parameter chosen to regulate the speed of the transfer. By engineering the nearest-neighbor couplings in accordance with $J_n = (\lambda/2)\sqrt{n(N-n)}$ where $J_n$ is the coupling between spins $n$ and $n+1$, we create a Hamiltonian with $J_n$ in the $(n,n+1)$ and $(n+1,n)$ positions for $n = 1, 2, \dotsc, N-1$ and zero otherwise.

The dynamics governing the system are given by the autonomous system
\begin{equation}\label{eq:perfect}
  \ket{\dot{\psi}(t)} = -jH\ket{\psi(t)}, \quad \ket{\psi(0)} = \ket{\psi_{\IN}}
\end{equation}
with solution $\ket{\psi(t)} = e^{-jHt}\ket{\psi_{\IN}}$. Since the overlap $\bra{\psi_{\OUT}} \psi(t) \rangle$ is complex, we transform this to the Bloch formulation to retain congruence with the previous sections and ensure a real fidelity and complementary error.

We use the generalized Gell-Mann basis~\cite{Bertlmann2008} of traceless, Hermitian $N \times N$ matrices for $\sigma_1$ through $\sigma_{N^2-1}$ with $\sigma_{N^2}$ as described in Section~\ref{qubit}. Applying~\eqref{eq:bloch1} and~\eqref{eq:bloch3} to the system of~\eqref{eq:perfect}, we get
\begin{subequations}\label{eq:bloch}\begin{align}
  \dot{r}(t) &= Ar(t), \quad r_{\IN} = r(0), \\
  r(t) &= e^{At}r_{\IN}, \\
  e(t) &= c r(t).
\end{align}\end{subequations}
$r_{\IN}$ is the Bloch-transformed version of $\rho_{\IN} = \ket{\psi_{\IN}}\bra{\psi_{\IN}}$, $r_{\OUT}$ is the transformed version of $\rho_{\OUT} = \ket{\psi_{\OUT}} \bra{\psi_{\OUT}}$, $c$ produces the error from the current state in the same manner as Section~\ref{qubit}, and $A$ is the Bloch-transformed Hamiltonian.

 \subsection{Log-Sensitivity of the Error}

With the purely coherent dynamics of~\eqref{eq:perfect}, perturbations of the Hamiltonian map linearly to the Bloch formulation via~\eqref{eq:bloch1}. For an $N$-chain we consider the $N-1$ possible perturbations to coupling strengths. These are structured as $S_n$, an $N \times N$ matrix with zeros everywhere save for ones in the $(n+1,n)$ and $(n,n+1)$ positions. This is then mapped to an $N^2 \times N^2$ matrix in the Bloch formulation via~\eqref{eq:bloch1} with $S$ interchanged with $jH$.

In the Bloch formulation, the matrix $A$ has $N$ eigenvalues at zero and the remaining $N^2-N$ eigenvalues in purely imaginary complex conjugate pairs. In accordance with Theorem~\ref{thm: periodic}, we expect the log-sensitivity to exhibit oscillations of increasing magnitude that achieve local maxima with a period given the fundamental frequency of the set $\{\omega_n\}$; more general chains would show aperiodicity~\cite{QuRecThm}.

Given that $A$ is a normal matrix, we decompose it as $A = V \Lambda V^{\dagger}$ where $VV^{\dagger} = I$. Retaining the same notation as in Section~\ref{math}, we have $z_n = c v_n$ and $w_l =  v_l^{\dagger} r_0 $ where $v_k$ is the $n$-th column of $V$ and $v_l^{\dagger}$ is the conjugate transpose of the $l$-th column of $V$. We then compute $s(\xi_0,t)$ in accordance with~\eqref{eq:limit} and~\eqref{eq:log-sensitivity}.

In Figs.~\ref{fig:perfect1} and~\ref{fig:perfect2} we show the behavior of a two-chain with $\lambda = \pi/5$ and perturbation on the coupling between the two spins with nominal value $J_1 = \pi/10$. Fig.~\ref{fig:perfect3} depicts the same for a chain of three spins and perturbation on the $2$--$3$ coupling with nominal value $J_2 = \sqrt{2} \pi/10$. In both cases, $s(\xi_0,t)$ does not have a defined limit, demonstrates the periodic spikes at times of perfect state transfer, and grows with time in accordance with Theorem~\ref{thm: periodic} and the accompanying corollary.
\begin{figure}[t]
\centering
\includegraphics[width=\columnwidth]{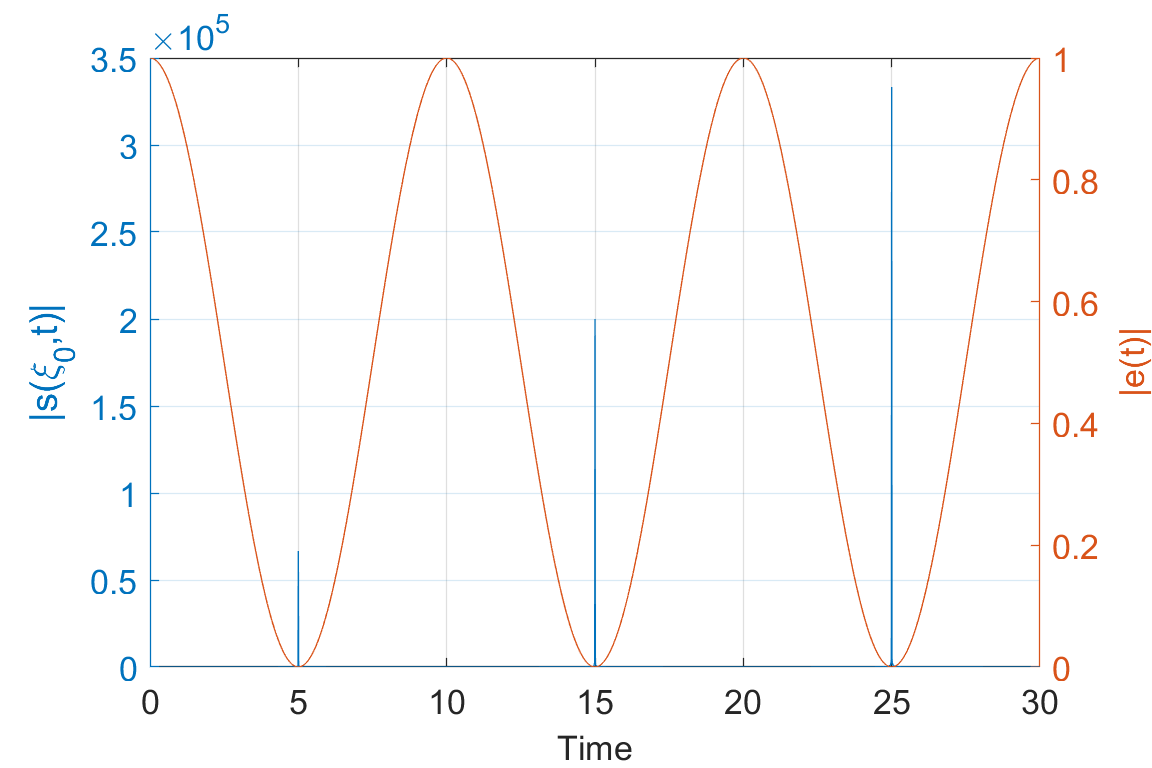}
\caption{Plot of $\left|s(\xi_0,t) \right|$ and $|e(t)|$ on a linear scale for a two-chain with perturbation on the $J_1$ coupling.}
\label{fig:perfect1}
\end{figure}

\begin{figure}[t]
\centering
\includegraphics[width=\columnwidth]{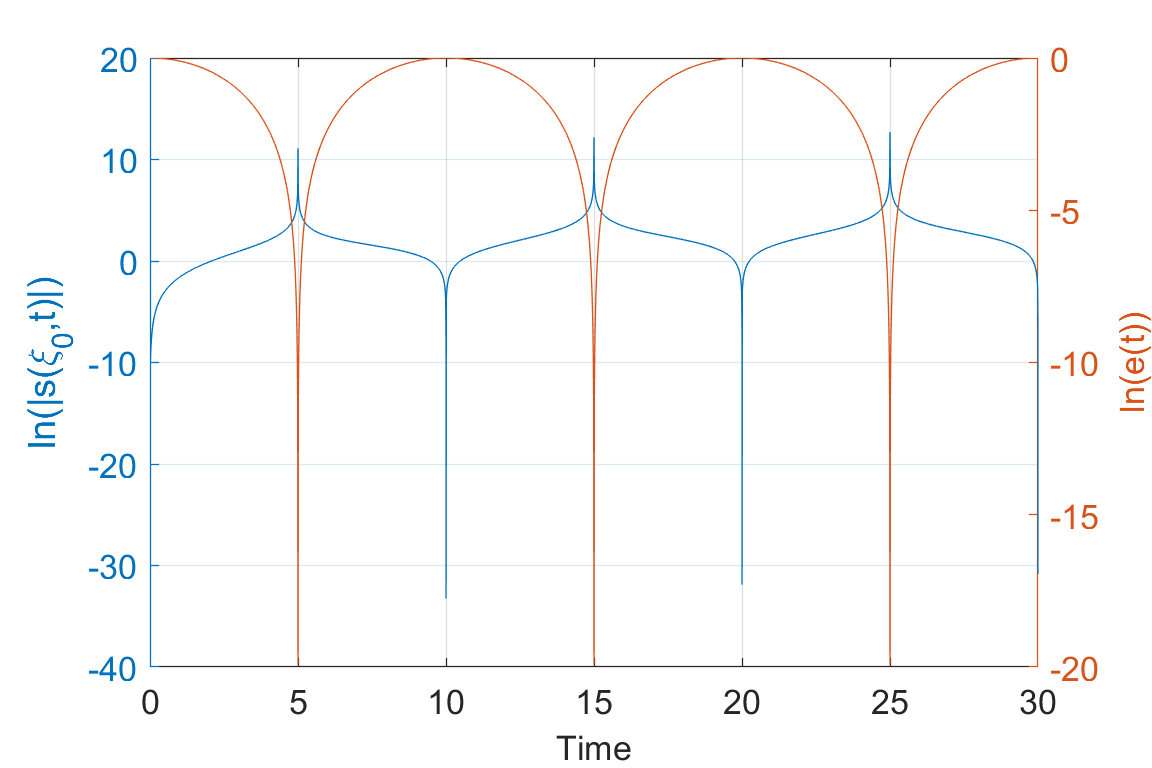}
\caption{Plot of $\left|s(\xi_0,t) \right|$ and $|e(t)|$ on a log scale for a two-chain with perturbation on the $J_1$ coupling.}\label{fig:perfect2}
\end{figure}

\begin{figure}[t]
\centering
\includegraphics[width=\columnwidth]{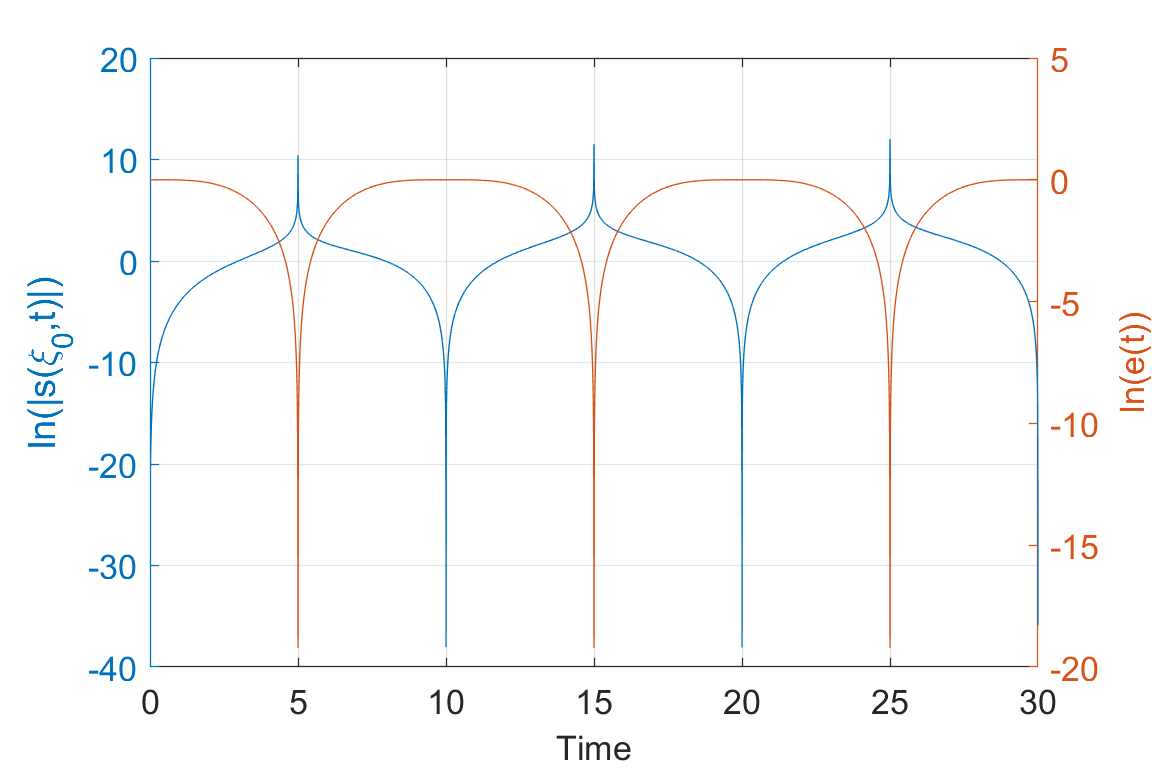}
\caption{Plot of $\left|s(\xi_0,t) \right|$ and $|e(t)|$ on a log scale for a three-chain with perturbation on the $J_2$ coupling.}\label{fig:perfect3}
\end{figure}

Additionally, there is no contradiction with the assertion of earlier work~\cite{Schirmer2016}, that under the conditions for perfect state transfer (superoptimality) the \emph{sensitivity} of the error goes to zero. Though,~\cite{Schirmer2016} states this characteristic holds for rings, we can see that it also holds for the chains engineered for perfect state transfer considered here. For $N=2$, calculation of the output matrix in accordance with Section~\ref{qubit} yields $c = \begin{bmatrix}0 & 0 & \frac{1}{\sqrt{2}} & \frac{1}{\sqrt{2}} \end{bmatrix}$ with $r(t) = \begin{bmatrix} 0 & \frac{1}{\sqrt{2}}\sin(\frac{\pi}{5}t) & \frac{1}{\sqrt{2}}\cos(\frac{\pi}{5}t) & \frac{1}{\sqrt{2}}\end{bmatrix}^{T}$. The resulting error is $e(t) = cr(t) = \frac{1}{2} \left( 1 + \cos(\frac{\pi}{5}t) \right)$, and the \emph{sensitivity} is $\partial e(t)/\partial \xi = t \sin(\pi/5 t)$. Thus, $\partial{e(t)}/\partial \xi = 0$ if $t = t_n = 5(2n + 1)$. Regarding the \emph{log-sensitivity,} we have
\begin{equation}\label{eq:sensitivity limit}
  \lim_{t \rightarrow t_n} \frac{\partial e(t)}{\partial \xi}\frac{\xi}{e(t)} \bigg|_{\xi_0} = \frac{0}{0}.
\end{equation} 
Applying L'Hopital's rule to this indeterminate form yields
\begin{equation}\label{eq:lhopital}
  \lim_{t \rightarrow t_n} \frac{\frac{\pi}{10} \left( \sin(\frac{\pi}{5}t) + \frac{\pi}{5}t\cos(\frac{\pi}{5}t \right)}{-\frac{\pi}{10}\sin(\frac{\pi}{5}t)} = \frac{\pi(2n+1)}{0},
\end{equation}
which is consistent with the expected result for the log-sensitivity. A similar argument holds for the case of $N=3$ with perturbation on the $2$--$3$ coupling. For this scenario,
\begin{align*}
  \frac{\partial e(t)}{\partial \xi} &= -\frac{\sqrt{2}}{4} t \sin\left(\frac{\pi}{5}t\right) + \frac{\sqrt{2}}{8} t \sin \left(\frac{2\pi}{5} t \right),\\
  e(t) &= \frac{5}{8} + \frac{1}{2} \cos \left( \frac{\pi}{5} t \right) - \frac{1}{8} \cos \left( \frac{2 \pi}{5} t \right).
\end{align*}
Applying the same procedure as~\eqref{eq:sensitivity limit} and~\eqref{eq:lhopital} to the equations above yields the same result: the \emph{sensitivity} $\frac{\partial e(t)}{\partial \xi} \rightarrow 0$ as $t \rightarrow t_n$. However, the \emph{log-sensitivity} goes to infinity at each $t_n$ determined by the fundamental frequency $\omega = \frac{\pi}{5}$ of the pair $\{\frac{\pi}{5},\frac{2\pi}{5} \}$.
Furthermore, we note the trade-off between the error and the log-sensitivity -- the periods of near-zero error (near perfect fidelity) correspond to those with the greatest log-sensitivity.  Table~\ref{table1} shows the trade-off between log-sensitivity and fidelity numerically.

\begin{table}[t]
\caption{Fidelity vs log-sensitivity for two- and three-chains at first fidelity maximum $t = 5$ under coupling perturbation.}\label{table1}
\centering
\begin{tabular}{|| r@{.}l | r@{.}l || r@{.}l | r@{.}l ||} \hline
  \multicolumn{4}{||c||}{$N=2$, $1$--$2$ Coupling} & \multicolumn{4}{c||}{$N=3$, $2$--$3$ Coupling} \\
  \hline
  \multicolumn{2}{||c|}{Fidelity} & \multicolumn{2}{c||}{$|s(\xi_0,t) |$} & \multicolumn{2}{c|}{Fidelity} & \multicolumn{2}{c||}{$|s(\xi_0,t) |$} \\
  \hline
  1&0     & 66664&00   & 1&0     & 24998&00\\\hline
  0&9999  &   311&95   & 0&9999  &   220&72\\\hline
  0&99899 &    97&271  & 0&99899 &    69&195\\\hline
  0&98999 &    29&264  & 0&98996 &    21&079\\\hline
  0&90001 &     7&4949 & 0&90008 &     5&6196\\\hline
\end{tabular}
\end{table}

\section{Discussion and Conclusions}

We have shown that the log-sensitivity of the error can be reliably computed from a time-domain perspective. More importantly, this robustness measure is applicable to a spectrum of classical and quantum systems and exhibits the same key characteristic: as performance increases (error gets smaller) the measure of performance is more sensitive to parameter variation.

Within the context of previous work on the robustness of energy landscape shaping, we begin to see the pattern of results regarding error versus log-sensitivity unified under this time-domain specification. In~\cite{Schirmer2016}, we demonstrate that when conditions for superoptimality prevail in a ring, the sensitivity to parameter variation vanishes. As shown here in Section~\ref{state transfer} we obtain the same for chains. More importantly, while the sensitivity vanishes, the log-sensitivity diverges at the instants of perfect state transfer, as predicted by the theory. In~\cite{Schirmer2018}, the trend of lower fidelity controllers exhibiting lower sensitivity to decoherence was observed by calculating the derivative of the error through a finite difference approximation, in agreement with the analytical methods presented here. While the overall trend suggested discordance between lower error and lower sensitivity, the trend was far from uniform. In the present paper, the sharp increase of the log-sensitivity as the fidelity approaches its maximum is seen in Table~\ref{table1} and suggests a justification for this variability of log-sensitivity for extremely high-fidelity controllers. Taken together, this indicates that designing controllers with an acceptable error and guaranteed robustness margin is possible.

Next, we note that the methodology of this paper is applicable to both open and closed quantum systems. Previous work on the application of classical robust control techniques to quantum systems has focused on \emph{open} quantum systems (i.e., those with dissipative behaviors that produce left-half plane poles), such as the $\mu$-analysis of~\cite{Schirmer2022} or a classically-inspired stability margin~\cite{CDC2022}. While~\cite{ IRP_quantum_survey_as_published, Petersen2013, IRP_old, coherent_h-infinity} apply $H^{\infty}$ methods with great success to a specific class of optical systems described by linear quantum stochastic differential equations, dissipation is still a necessary component to ensure application of the bounded real lemma. While~\cite{Koswara_2021} concludes that a tradeoff between performance and robustness is necessary in closed quantum systems, the approach is purely stochastic, based on the expected value versus the variance of the optimization functional. 

In terms of future work, while we have shown a classical trend between the error and log-sensitivity, we have not shown any guaranteed robustness bounds along the lines of the identity $S(j\omega) + T(j\omega) = I$. 
Secondly, the behavior of the log-sensitivity at the transition from complex to repeated, real eigenvalues still requires attention.  Furthermore, to bolster applicability to quantum networks, extension to non-linear and non-autonomous systems with time-varying controls and non-linear performance measures such as concurrence~\cite{EntanglementConcurrence} to measure entanglement is necessary.

\bibliographystyle{ieeetr}
\bibliography{export}

\appendix

We use the Jordan decomposition to derive a general formula for the matrix derivative. Any square matrix $A$ of dimension $N$ over the field of complex numbers is similar to a Jordan normal form $A = M J M^{-1}$, where $J$ is the direct sum of $\ell$ Jordan blocks, each with dimension $n_m$ so that $\sum_{m=1}^\ell n_m = N$~\cite{Belevitch}. There are two cases. If all Jordan blocks have dimension $1$ then $A$ is said to be diagonalizable. If there are eigenvalues whose geometric multiplicity is smaller than their algebraic multiplicity then the Jordan decomposition has nontrivial Jordan blocks. Since the diagonalizable matrices form an open and dense subset in the space of matrices, this case is generic.

\subsection{Generic Case: Diagonalizable $A$}\label{ss: generic}

When $A_0$ is diagonalizable, $A_0 = M \Lambda M^{-1}$, $e^{A_0t} = M e^{t \Lambda} M^{-1}$ where $\Lambda$ is a diagonal matrix of eigenvalues $\lambda_m$. We then have~\cite{Havel1995}
\begin{equation}
  \frac{\partial}{\partial \xi} e^{At} = M (\bar{S} \odot \Phi(t) ) M^{-1},
\end{equation}
where $\bar{S} = M^{-1}SM$, $\odot$ is the Hadamard product, and the elements $\phi_{mn}(t)$ of $\Phi(t)$ are as defined in~\eqref{eq: phi}. Let $\{ \hat{e}_m \} \in \mathbb{R}^{N}$ be the set of natural basis vectors for $\mathbb{R}^{N}$ with a $1$ in the $m$-th position and zeros elsewhere. Define a basis for the $N \times N$ space of linear operators on $\mathbb{R}^{N}$ as $\Pi_{mn} = \hat{e}_m \hat{e}_n^{T}$.  Then $\frac{\partial e^{At}}{\partial \xi} = M X(t) M^{-1}$ with
\begin{multline}\label{eq:diag}
  X(t) =  \sum_{m,n} \bar{s}_{mn}\phi_{mn}(t) \Pi_{mn} = \\
   \sum\limits_{\substack{m,n \\ \lambda_m=\lambda_n}} \bar{s}_{mn}te^{\lambda_m t} \Pi_{mn} + \sum\limits_{\substack{m,n \\ \lambda_m \neq \lambda_{n}}} \bar{s}_{mn} \frac{e^{\lambda_m t} - e^{\lambda_n t}}{\lambda_m - \lambda_n} \Pi_{m n} 
\end{multline}
where $\bar{s}_{mn}$ is the element in the $m,n$ position of $\bar{S}$.

\subsection{Non-Generic Case: Non-Trivial Jordan Decomposition}\label{ss: jordan_case}

Consider the case of algebraic multiplicity $\ell$ in the dominant eigenvalue $\lambda_1$ with geometric multiplicity $1$ and the remaining $N-\ell$ eigenvalues distinct. Note that application of the following is restricted to the case where a versal deformation of the Jordan normal form $J$ in terms of $\xi$ does not admit a bifurcation in the spectrum of $A$~\cite{versal_deformation,arnold}, which would break the degeneracy and default to the generic case of Section~\ref{ss: generic}. Write the matrix exponential $Me^{Jt}M^{-1}$ as
\begin{equation}
  M \left( \sum\limits_{m=1}^N e^{\lambda_m t} \Pi_{mm} + \sum\limits_{p=1}^{\ell-1} \sum\limits_{q = p+1}^{\ell} e^{\lambda_1 t} \frac{t^{(q-p)}}{(q-p)!}\Pi_{pq} \right) M^{-1}.
\end{equation}
Let $\lambda_m = \lambda_1$ for $m = 1$ to $\ell$ so that the first eigenvalue not identical to $\lambda_1$ is $\lambda_{\ell+1}$.  In accordance with Eq.~\eqref{eq: matrix_deriv}, we have
\begin{multline}
   Me^{J(t-\tau)}\bar{S}e^{J\tau}M^{-1} \\ 
   = M \left( \sum\limits_{m,n = 1}^{N} e^{\lambda_m t} e^{(\lambda_n - \lambda_m ) \tau} \Pi_{mm}\bar{S} \Pi_{nn} +  \right. \\
   \sum\limits_{m = 1}^N \sum\limits_{r = 1}^{\ell-1} \sum\limits_{n=r+1}^\ell e^{\lambda_m t} e^{(\lambda_1 - \lambda_m) \tau} \frac{\tau^{(n-r)}}{(n-r)!} \Pi_{mm}\bar{S}\Pi_{rn} +\\
   \sum\limits_{m = 1}^{N} \sum\limits_{p=1}^{\ell-1}\sum\limits_{q=p+1}^\ell e^{\lambda_1 t} e^{(\lambda_m - \lambda_1) \tau} \frac{(t - \tau)^{(q-p)}}{(q-p)!}\Pi_{pq}\bar{S}\Pi_{mm} + \\
   \left.  \sum\limits_{\substack{p = 1\\  r=1}}^{\ell-1} \sum\limits_{\substack{q = p +1\\ n = r + 1}}^\ell e^{\lambda_1 t} \frac{(t-\tau)^{(q-p)}\tau^{(n-r)}}{(q-p)! (n-r)! }\Pi_{pq}\bar{S}\Pi_{rs} \right) M^{-1} \\
   = M \left[ \mathcal{X}_1(t) + \mathcal{X}_2(t) + \mathcal{X}_3(t) + \mathcal{X}_4(t) \right] M^{-1}. 
\end{multline}
Calculation of $M \left[ \int_{0}^{t} e^{J(t-\tau)}\bar{S}e^{J \tau} d \tau \right] M^{-1}$ produces the following:

Firstly, $\int_{0}^{t} \mathcal{X}_1(\tau) d \tau = X_1(t)$ produces the same result as a fully diagonalizable matrix with $\ell$ repeated eigenvalues $\lambda_1$ with solution given by the term in parentheses in~\eqref{eq: diagonal_repeats}. Secondly,
\begin{multline} \label{eq:B int}
    \int_{0}^t \mathcal{X}_2(\tau) d \tau
    = \sum\limits_{m = 1}^{\ell} \sum\limits_{r=1}^{\ell-1} \sum\limits_{n=r+1}^\ell e^{\lambda_1 t} \frac{t^{(n - r +1)}}{(n - r + 1)!} \bar{s}_{mr} \Pi_{mn}\\
    + \sum\limits_{m = \ell +1}^{N}\sum\limits_{r=1}^{\ell-1}\sum\limits_{n = r+1}^{\ell}  \left( \sum_{i = 0}^{n-r} \frac{(-1)^{i} e^{\lambda_1 t}(n-r)! t^{(n-r-i)}}{(n-r-i)!(\lambda_1 - \lambda_m)^{(i+1)}} \right.\\
    \left. + \frac{(-1)^{(n-r+1)} (n-r)! e^{\lambda_m t}}{(\lambda_1 - \lambda_m)^{(n-r+1)}} \right) \bar{s}_{mr} \Pi_{mn} = X_2(t).
\end{multline}
Likewise, 
\begin{multline} \label{eq:C int}
    \int_{0}^t \mathcal{X}_3(\tau) d \tau
    =\sum\limits_{m = 1}^{\ell} \sum\limits_{p=1}^{\ell-1} \sum\limits_{q=p+1}^\ell e^{\lambda_1 t} \frac{t^{(q - p +1)}}{(q-p + 1)!} \bar{s}_{q m} \Pi_{p m}\\
    + \sum\limits_{m = \ell+1}^{N}\sum\limits_{p=1}^{\ell-1}\sum\limits_{q = p+1}^{\ell}  \left( \sum_{i = 0}^{q-p}\frac{(-1)^{i} e^{\lambda_1 t} (q-p)! t^{(q-p-i)}}{(q-p-i)!(\lambda_1 - \lambda_m)^{(i+1)}} \right. \\
    \left. + \frac{(-1)^{(q-p+1)} (q-p)! e^{\lambda_m t}}{(\lambda_1 - \lambda_m)^{(q-p+1)}} \right) \bar{s}_{q m}\Pi_{p m} = X_3(t).
\end{multline}
Integrating on $\mathcal{X}_4(t)$ provides 
\begin{multline}
  \int_0^t \mathcal{X}_4(\tau) d \tau \\ 
  = \sum\limits_{\substack{p =1 \\ r=1}}^{\ell-1} \sum\limits_{\substack{q = p +1 \\ n = r+1}}^{\ell} \frac{e^{\lambda_1 t}t^{(q-p + n-r + 1)}}{(q-p + n - r +1)!} \Pi_{pq}\bar{S} \Pi_{rn} = X_4(t).
\end{multline}
We thus have $\frac{\partial e^{At}}{\partial \xi} = M X M^{-1}$ with
\begin{equation}\label{eq:non-diag}
  X(t) = \sum_{m=1}^4 X_m(t).
\end{equation}

\end{document}